\documentclass[a4paper,11pt]{article}
\pdfoutput=1 

\usepackage{jheppub} 

\usepackage[T1]{fontenc} 
\usepackage{subfig}
\usepackage{enumerate}
\usepackage{enumitem}
\newcommand{\llangle}{\langle\!\langle}
\newcommand{\rrangle}{\rangle\!\rangle}

\usepackage{graphicx}
\usepackage{color}
\usepackage[export]{adjustbox}
\usepackage{amssymb}
\usepackage{physics}
\usepackage{amsmath}
\usepackage{amsthm}
\usepackage{mathrsfs}
\usepackage{amsmath}
\usepackage{cleveref}

\DeclareMathOperator{\reg}{reg}

\renewcommand{\tilde}{\widetilde}
\renewcommand{\bar}{\overline}
\renewcommand{\S}{\mathcal{S}}
\newcommand{\A}{\mathcal{A}}
\newcommand{\B}{\mathcal{B}}

\renewcommand{\H}{\mathcal{H}}
\DeclareMathOperator{\Rep}{Rep}

\DeclareMathOperator{\id}{id}
\DeclareMathOperator{\triv}{triv}

\newcommand{\rightZ}[1]{\overset{\lower2px\hbox{$\to$}}{#1}}
\newcommand{\leftZ}[1]{\overset{\lower2px\hbox{$\leftarrow$}}{#1}}

\newcommand{\cO}{\mathcal{O}}
\newcommand{\w}{\widehat}

\newtheorem{theorem}{Theorem}[section]
\newtheorem{lemma}[theorem]{Lemma}
\newtheorem{proposition}[theorem]{Proposition}
\newtheorem{corollary}[theorem]{Corollary}

\numberwithin{question}{section}

\numberwithin{question}{section}

\theoremstyle{definition}
\newtheorem{example}{Example}[section]

\theoremstyle{remark}

\usepackage{bm}
\usepackage{tikz}
\usepackage{tikz-3dplot}
\usepackage{verbatim}
\usepackage{mathrsfs}
\usepackage{physics}

\usepackage[dvipsnames]{xcolor}
\usepackage[most]{tcolorbox}

\newtcolorbox{resultbox}{
  colback=MidnightBlue!5!white,
  colframe=MidnightBlue!75!black,
  breakable
}
\newtcolorbox{partialresultbox}{
  colback=Gray!5!white,
  colframe=Gray!75!black,
  breakable
}

\DeclareMathOperator{\End}{End}

\newcommand{\C}{\mathbb{C}}

\newcommand{\cA}{\mathcal{A}} 

\newcommand{\cH}{\mathcal{H}}

\usepackage{tikz-cd}
\usetikzlibrary{decorations.pathmorphing}
\usetikzlibrary{decorations.pathreplacing,decorations.markings}

\tikzset{
  on each segment/.style={
    decorate,
    decoration={
      show path construction,
      moveto code={},
      lineto code={
        \path [#1]
        (\tikzinputsegmentfirst) -- (\tikzinputsegmentlast);
      },
      curveto code={
        \path [#1] (\tikzinputsegmentfirst)
        .. controls
        (\tikzinputsegmentsupporta) and (\tikzinputsegmentsupportb)
        ..
        (\tikzinputsegmentlast);
      },
      closepath code={
        \path [#1]
        (\tikzinputsegmentfirst) -- (\tikzinputsegmentlast);
      },
    },
  },
  mid arrow/.style={postaction={decorate,decoration={
        markings,
        mark=at position .5 with {\arrow[#1]{stealth}}
      }}},
}

\usepackage{xcolor}
\usepackage[framemethod=default]{mdframed}
\usepackage{lipsum} 

\title{Algebraic locality and non-invertible Gauss laws}
\author[a]{Nicholas Holfester}
\author[b]{and Jonathan Sorce}
\affiliation[a]{Leinweber Institute for Theoretical Physics, University of Chicago}
\affiliation[b]{Princeton Gravity Initiative, Princeton University}

\abstract{
We study algebraic locality principles on a $2+1$D closed lattice in the presence of a Gauss law for a non-invertible symmetry. Prior work in arXiv:2509.03589 showed that when enforcing the Gauss law of an invertible symmetry, the principle of ``Haag duality'' is preserved exactly, and ``disjoint additivity'' is preserved after appropriate treatment of discreteness artifacts.
Here we show that for a large class of non-invertible on-site symmetries, Haag duality is preserved exactly only for sufficiently nice, ``cuspless'' regions.
For cusped regions, we instead have a weak form of Haag duality that requires adding a collar.
Our results apply to double models with a purely magnetic constraint, and to the more general framework of constraints induced by the on-site action of a Hopf algebra. In particular, we treat a class of extended string-net models explicitly.
We also demonstrate disjoint additivity for double models based on a group, and a weakened form of disjoint additivity in the setting of a general Hopf algebra.
}

\begin{document}
\maketitle

\section{Introduction}

Since \cite{Haag-Kastler}, it has been understood that operator algebras are useful for characterizing locality in quantum systems.
The basic idea is to organize the degrees of freedom into local subregions $R$ carrying operator algebras $\A(R)$, and then to study the structure of the system via properties of the map $R \mapsto \A(R).$
Many interesting insights into both continuum and discrete quantum theories have been obtained using this framework; for a very incomplete list, with a focus on recent results, see \cite{DHR-1, DHR-2, Fredenhagen:braids-1, Fredenhagen:braids-2, Longo:subfactors, Naaijkens:index, Naaijkens:subfactors, Casini:order, Casini:completeness, Jones:boundary, Benedetti:modular, Shao:additivity, Evans:fusion, Harlow:disjoint-additivity, vanLuijk:cones, Casini:DHR}.

One setting in which local operator algebras have an interesting structure is that of a constrained lattice system.
An \textit{unconstrained} system on a finite lattice is simply a tensor product Hilbert space $\H = \otimes_{j} \H_{j};$ the local operator algebra associated with a subset $R$ of tensor factors is just \begin{equation}
    \A(R) = \B(\H_R) = \B(\otimes_{j \in R} \H_j),
\end{equation}
where $\B(\cdot)$ denotes the space of bounded operators.\footnote{On infinite lattices, there are interesting infrared effects involved in the construction of operator algebras for non-compact regions; for a few nice papers on this subject, see \cite{Naaijkens:index, vanLuijk:cones, Ogata:HD}.}
In the presence of a \textit{constraint}, like a gauge symmetry, one restricts to a subspace $\H_0 \subseteq \H$ of ``physical states,'' and this induces a corresponding constraint on the local operator algebras; under such a constraint, the structure of the local algebras can become richer.

A large class of interesting lattice constraints is provided by the setting where $\H_0$ is the $+1$ eigenspace of a local family of commuting projectors.
This includes the setting of lattice gauge theory, where an on-site group action is gauged to produce a constrained space $\H_0$ of gauge-invariant states.
In the recent paper \cite{Harlow:disjoint-additivity}, it was shown that on finite lattices, the algebraic locality principles of ``Haag duality'' and ``disjoint additivity'' are satisfied in the constraint subspace corresponding to a group-like gauge symmetry.
The point of the present paper is to generalize this analysis to a broader class of constraints implemented by gauging non-invertible symmetries. 

The simplest setting we study is that of a quantum double model \cite{Kitaev:double} where the magnetic constraint is imposed exactly, while the electric constraint is imposed energetically.
By contrast, the analysis of \cite{Harlow:disjoint-additivity} corresponds to the dual case where the electric constraint is exact.
We explain in section \ref{sec:background} how to think of the magnetic constraint as the singlet projection of a $\Rep(G)$ gauge symmetry, which allows us to conceptualize the magnetic constraint of a double model based on $G$ as the electric constraint of a double model based on $\Rep(G).$
This observation leads us to study algebraic locality in general double models based on Hopf algebras \cite{Buerschaper_2013, Buerschaper:2010yf, Yan_2022}, which provide a large class of Gauss law constraints based on non-invertible symmetries.\footnote{Note that our usage of this term refers to the gauging of a non-invertible zero-form symmetry; for another notion of ``non-invertible Gauss law'' based on a global non-invertible one-form symmetry, see \cite{Choi:axions}.}

The plan of the paper, including a summary of our main results, follows.
\begin{itemize}
    \item In section \ref{sec:background}, we review the definitions of Haag duality and disjoint additivity, and give a general definition of local operator algebras corresponding to a constraint subspace $\H_0 \subseteq \H.$
    We also review the quantum double model, and explain how to think of the magnetic constraint as a Gauss law corresponding to a $\Rep(G)$ gauge symmetry.
    \item In section \ref{sec:counterexample}, we show by explicit counterexample that Haag duality is not satisfied in general for local operator algebras under the magnetic constraint of a non-abelian double model.
    \item In section \ref{sec:cuspless-magnetic}, we show that under the magnetic constraint of a non-abelian double model, a weak form of Haag duality is always satisfied, where ``weak'' means that it is sometimes necessary to add a small collar to the region under investigation.
    We show that for ``cuspless'' regions, Haag duality is exact; we also point out that on a trivalent lattice, every region is cuspless.
    We further show that in this setting, disjoint additivity holds when regions $R_1$ and $R_2$ are defined to be ``disjoint'' if no vertex is incident to edges in both $R_1$ and $R_2.$
    \item In section \ref{sec:hopf}, we introduce double model constraints based on a general Hopf algebra, and show that the conclusions about Haag duality are exactly the same as in the nonabelian-magnetic case.
    We also demonstrate a weakened form of disjoint additivity that involves additional discreteness artifacts; however we note that we do \textit{not} find a counterexample to the stronger form of disjoint additivity, and we suspect that it holds even in this more general setting.
    \item Finally, in section \ref{sec:discussion}, we comment on potential generalizations to weak Hopf algebras and full Drinfel'd doubles, and we comment on the potential for a more detailed analysis of disjoint additivity in the general Hopf-algebraic setting.
\end{itemize}

\section{Background}
\label{sec:background}

This section provides background material.
Not all of it needs to be read to approach the rest of the paper.
Section \ref{sec:background-algebraic} reviews basic algebraic locality principles, and section \ref{subsec:qdouble} reviews the quantum double model; these sections are necessary to understand the Haag duality violations in section \ref{sec:counterexample}.
Section \ref{subsec:gauge_theory_rep} discusses the magnetic constraint of the double model in representation-theoretic terms; it is of conceptual importance for the main ideas of the paper, and provides useful tools for studying algebraic locality in section \ref{sec:cuspless-magnetic}.

\subsection{Algebraic preliminaries}
\label{sec:background-algebraic}

Given a quantum theory with a Hilbert space $\H,$ one is typically also given a collection of ``regions'' into which degrees of freedom can be localized.
These degrees of freedom are specified by listing the operators on $\H$ that can be accessed within a given region.
More formally, one has a ``net of algebras,'' which is just a map
\begin{equation}
    R \mapsto \A(R)
\end{equation}
that takes regions to sets of operators on $\H.$ To be called a net of local algebras, the above map is taken to have a few natural properties:
\begin{itemize}
    \item Each algebra $\A(R)$ is a von Neumann algebra, meaning it is a subalgebra of bounded operators on $\H$ that is complete with respect to a natural topology.\footnote{For detailed definitions see e.g. the review in \cite{Sorce:notes}.}
    An important property of von Neumann algebras is that if one defines the \textit{commutant} by\footnote{The symbol $\B(\H)$ denotes the set of bounded operators on $\H$. In finite dimensions, all linear operators are bounded.}
    \begin{equation}
        \A(R)' \equiv \{a' \in \mathcal{B}(\H) \text{ such that } [a', a] = 0 \text{ for all } a \in \A(R)\},
    \end{equation}
    then a von Neumann algebra is equal to its own double commutant, $\A(R) = \A(R)''.$
    \item Larger regions have larger algebras, i.e., one has the implication
    \begin{equation} \label{eq:nesting}
        (R_1 \subseteq R_2) \quad \Rightarrow \quad (\A(R_1) \subseteq \A(R_2)).
    \end{equation}
    \item
    One assumes that there is a ``complement operation'' which associates, to any region $R$, a complementary region $R'.$
    For bosonic algebras --- which will be the only ones we consider in the present paper --- one assumes the microcausality inclusion
    \begin{equation} \label{eq:microcausality}
        \A(R') \subseteq \A(R)'.
    \end{equation}
\end{itemize}

In specific quantum theories of interest, the ``set of regions'' may be, for example: subsets of a lattice; open sets on a time slice with sufficiently smooth boundary conditions; causally complete codimension-zero spacetime regions; etc.
We refer the reader to \cite[section 2 and appendix A]{Harlow:disjoint-additivity} for a general discussion.
In the present paper, we will always be working on closed lattices, and our regions will be subsets of those lattices.

The conditions on $\A(R)$ listed above are very weak.
There are additional completeness relations that are typically studied in practice.
The important ones for the present paper are:
\begin{itemize}
    \item A net of local algebras satisfies \textit{Haag duality} for the region $R$ if one has
    \begin{equation} \label{eq:HD}
        \A(R') = \A(R)',
    \end{equation}
    i.e., if inequality \eqref{eq:microcausality} is saturated.
    \item A net of local algebras satisfies \textit{additivity} for regions $R_1$ and $R_2$ if one has
    \begin{equation} \label{eq:additivity}
        \A(R_1 \cup R_2) = \A(R_1) \vee \A(R_2),
    \end{equation}
    where the symbol $\vee$ denotes the algebraic union, i.e., the right-hand side of the above equation is the smallest von Neumann algebra containing both $\A(R_1)$ and $\A(R_2).$

    Note that the $\supseteq$ inclusion in equation \eqref{eq:additivity} is guaranteed by the nesting relation \eqref{eq:nesting}, since one has $R_1 \subseteq R_1 \cup R_2$ and $R_2 \subseteq R_1 \cup R_2.$
\end{itemize}
By studying Haag duality and additivity for nets associated to a quantum theory, one can uncover interesting things about observables of the theory. For example, if one considers the local operator algebra generated by the algebras assigned to simply connected regions --- the so called ``additive net'' $\A_{\text{add}}(R)$ --- then violations of Haag duality can diagnose the existence of higher-form symmetries \cite{Casini:order, Casini:completeness}.
On the other hand, if one studies concrete examples in which the full spectrum of extended operators is understood, and constructs the net $\A_{\text{complete}}(R)$ that includes all of these operators, then the existence of higher-form symmetries instead leads to violations of additivity \cite{Shao:additivity, Harlow:disjoint-additivity}. 

To study the latter case in greater detail, the paper \cite{Harlow:disjoint-additivity} introduced a notion of \textit{disjoint} additivity.
Disjoint additivity does not require equation \eqref{eq:additivity} to hold for arbitrary regions $R_1$ and $R_2,$ but only for those that satisfy an appropriate disjointness relation
\begin{equation}
    \bar{R_1} \cap \bar{R_2} = \varnothing.
\end{equation} 
This expression is heuristic; defining disjoint additivity in any concrete example requires specifying what one means by the closure of a region, and also what one actually means by intersection.
On the lattice, closure usually means adding some finite-radius collar to the region, and intersection is just the intersection of sets.
In the relativistic continuum, closure means the topological closure, but when one says that the closed regions have ``no intersection'' one really means that they are spacelike separated from one another.
More detail on this general definition can be found in \cite[section 2]{Harlow:disjoint-additivity}.

In this paper we will be interested in studying Haag duality and additivity for constrained subspaces of lattice tensor product Hilbert spaces. Physically, these will correspond to subspaces having a (possibly non-invertible) Gauss law enforced.

For a general lattice Hilbert space $\H$, consider some sub-Hilbert space $\H_0$, and let $\Pi_0$ denote the orthogonal projector onto the subspace. It makes sense to talk about the set of ``subspace-preserving operators'' as the set of lattice operators $L$ satisfying 
\begin{equation}
    L \Pi_{0}
        = \Pi_{0} L \Pi_{0}, \quad L \in \B(\cH).
\end{equation}
These are operators that map states in $\cH_0$ to states in $\cH_0$; no claims are made about what they do to other states in $\H$.
One can then generally define a net of local algebras within the subspace via the formula 
\begin{equation} \label{eq:constrained-algebras}
    \A_{0}(R)
        \equiv \{ \Pi_{0} a \Pi_{0} \text{ such that } a \in \A(R) \text{ and } a \Pi_{0} = \Pi_{0} a \Pi_{0}\} \subseteq \B(\cH_0).
\end{equation}
Note that these operators act on the subspace $\H_0$ considered as its own Hilbert space; they are \textit{not} thought of as acting on the full Hilbert space $\H$.
This is important when discussing Haag duality, since commutants are taken within this smaller space.
Equation \eqref{eq:constrained-algebras} may be thought of intrinsically within the constraint space as saying,
\begin{quote}
    ``each constrained algebra $\A_0(R)$ consists of the operators that can be lifted, within the total Hilbert space, to constraint-preserving operators acting only on $R$.''
\end{quote}
This notion can be made more explicit by an alternative definition of $\A_0(R)$ as the image of the algebra homomorphism
\begin{equation}
    \rho: \cA_{\text{pres}}(R) \to \B(\cH_0),
\end{equation}
which restricts constraint-preserving operators to their action on the constraint subspace.

\subsection{The quantum double model}
\label{subsec:qdouble}

Here we review the quantum double model of \cite{Kitaev:double}.
We use different notation, and we use the ``dual frame'' as compared to that paper --- in other words, in our presentation, electric constraints are enforced at plaquettes and magnetic constraints are enforced at vertices.

Consider a finite group $G$. Let $\H_G$ denote the natural Hilbert space associated with the group algebra $\mathbb{C}[G]$:
\begin{equation}
    \H_{G}
        \equiv \text{span}\{|g\rangle \, |\, g \in G\},
\end{equation}
with
\begin{equation}
    \langle g | h \rangle = \delta_{g,h}.
\end{equation}
Given a finite lattice with oriented links, one specifies the Hilbert space of a $G$-double model by attaching one copy of $\H_G$ to each edge,
\begin{equation}
    \H
        = \bigotimes_{\text{edges $E$}} \H_{G, E}.
\end{equation}
At each plaquette $P$ of the lattice, and for each group element $g,$ one defines a group action $L_{P, g}$ that acts on the neighboring edges by left-multiplying each ``counterclockwise edge'' by $g$ and right-multiplying each ``clockwise edge'' by $g^{-1}$. This can be understood in diagrammatic terms as being implemented by the insertion of a non-contractible ``$g$-loop'' within the plaquette; see figure \ref{fig:gauge-transformation}.

\begin{figure}
    \centering
    \includegraphics{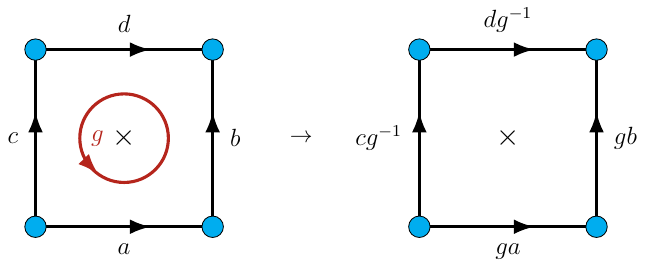}
    \caption{A $g$-loop acting on a plaquette. ``Fusing'' the $g$-loop outward causes all counterclockwise edges to be left-multiplied by $g,$ and all clockwise edges to be right-multiplied by $g^{-1}$.} 
    \label{fig:gauge-transformation}
\end{figure}

A state is said to satisfy the \textit{electric constraint} (or Gauss law) at plaquette $P$ if it is invariant under the action of all group elements at that plaquette.
It is said to satisfy the \textit{magnetic constraint} (or flux constraint) at a vertex $V$ if it respects the group fusion rules at that vertex, as in figure \ref{fig:fusion-constraint}.
More concretely, a collection of group elements surrounding a vertex satisfies the magnetic constraint if multiplying group elements in clockwise order around the vertex, and taking inverses each time there is an inward-pointing edge, the total result of this multiplication is the identity element of $G$. Both constraints are implemented by projection operators supported on vertex and plaquette respectively, defined as
\begin{equation}\label{eq:plaqutte_proj}
    \Pi_P = \frac{1}{|G|} \sum_{g \in G} \rho_P(g),
\end{equation}
where $\rho_P$ denotes the action of the $g$-loop as in figure \ref{fig:gauge-transformation}, and
\begin{equation}\label{eq:vertex_constraint_double}
    \Pi_V (\ket{g_1} \otimes \ket{g_2} \otimes \cdots \ket{g_n}) = \delta_{e,g_1g_2\cdots g_n} \ket{g_1} \otimes \ket{g_2} \otimes \cdots \ket{g_n},
\end{equation}
where the group elements are on outward oriented links as in figure \ref{fig:fusion-constraint}.

\begin{figure}
    \centering
    \includegraphics{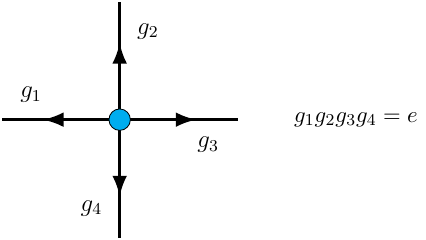}
    \caption{The magnetic constraint $g_1 g_2 g_3 g_4 = e$ at a four-valent vertex. If edge $j$ were directed inward instead of outward, one would replace $g_j$ with $g_j^{-1}.$}
    \label{fig:fusion-constraint}
\end{figure}

In constructing a quantum double model, the Hamiltonian is chosen to apply an energetic penalty to each plaquette where the electric constraint is violated and to each vertex where the magnetic constraint is violated. The ground space of this Hamiltonian is well known to be ``topologically ordered,'' and excitations above the ground space correspond to anyonic quasiparticles \cite{Kitaev:double}.
If instead the electric constraint is imposed exactly while the magnetic constraint continues to be imposed energetically, then this provides a model of pure lattice gauge theory with gauge group $G$ \cite{PhysRevD.11.395,PhysRevD.19.3715}. In this paper, we will be interested instead in exact imposition of the \textit{magnetic} constraint, with the electric constraint enforced energetically.
As we now explain, this model can be thought of as a kind of pure lattice gauge theory of a non-invertible $\Rep(G)$-symmetry.

\subsection{Magnetic constraints as a Rep$(G)$ Gauss law}\label{subsec:gauge_theory_rep}
As emphasized already in the original paper \cite{Kitaev:double}, the magnetic constraint of a double model permits a representation-theoretic description. Here we review this feature and its relation to a local $\Rep(G)$ symmetry \cite{Buerschaper_2009}. Via the connection to $\Rep(G)$ symmetry, we will then explain how to think of the magnetic constraint as a non-invertible Gauss law.

While the electric constraint of the double model is related to the group $G$, the magnetic constraint is related to the \textit{dual group algebra} $\C[G]^\vee$.\footnote{That an algebra can be a suitable analogue of the group is seen by recalling that a representation of a group is equivalent to a representation of its group algebra $\C[G]$. Thus, it is really $\mathbb{C}[G]$ that is being replaced by $\mathbb{C}[G]^{\vee}$ in this discussion.} This algebra has a basis labeled by elements $g \in G,$ and is represented on the link Hilbert spaces of a double model by rank-one projection operators,
\begin{equation}
    P_g = \ket{g}\bra{g}.
\end{equation}
These operators form a basis of $\C[G]^\vee$ and satisfy the multiplication rule
\begin{equation}
    P_g P_h = \delta_{g,h} P_h \quad g,h \in G.
\end{equation}

Any representation of the dual group algebra can be described equivalently as a $G$-graded vector space, which is a vector space $V$ equipped with a decomposition
\begin{equation}
    V = \bigoplus_{g \in G} V_g.
\end{equation}
In such a representation, the basis elements of $\mathbb{C}[G]^{\vee}$ are defined to act by projecting onto the corresponding labeled summand. One can always consider a one-dimensional representation labeled by a particular group element
\begin{equation}
    \delta_g \simeq \C,
\end{equation}
and $\{\delta_g\}$ is the set of irreducible representations of $\mathbb{C}[G]^{\vee}$.
On the link Hilbert space of the double model, the decomposition provided by the representation of the dual group algebra is simply the decomposition in the group basis:
\begin{equation}
    \cH_G = \bigoplus_{g \in G} (\cH_G)_g, \quad (\cH_G)_g = \text{span}\{\ket{g} \in \cH_G\}.
\end{equation}

Importantly, the dual group algebra $\mathbb{C}[G]^{\vee}$ naturally has the structure of a \textit{Hopf algebra}. We will review this structure in more detail in section \ref{sec:hopf} and appendix \ref{app:hopf}. For our immediate purposes, the Hopf algebra structure implies that
\begin{enumerate}
    \item representations admit tensor products, and
    \item there is a singlet representation.
\end{enumerate}
A basis element of the dual group algebra is defined to act on a tensor product of states as 
\begin{equation}\label{eq:dual_grp_alg_basis_act}
    P_g (\ket{v} \otimes \ket{w}) = \sum_{h,k \in G} \delta_{g,hk} (P_h\ket{v} \otimes P_k\ket{w}),
\end{equation}
which implies that irreducible representations have the product
\begin{equation}\label{eq:dual_rep_prod}
    \delta_g \otimes \delta_h \cong \delta_{gh}.
\end{equation}
This demonstrates that for a non-abelian group, reordering the tensor product can give non-isomorphic representations of $\C[G]^{\vee}$.\footnote{This is in contrast to representations of a group, which always satisfy $V \otimes W \cong W \otimes V$.} Therefore, when specifying a tensor product of states, a particular \textit{ordering} must also be specified.

The singlet representation is defined to be the irreducible representation labeled by the identity element. The basis elements of $\mathbb{C}[G]^{\vee}$ therefore act in this representation as
\begin{equation}
    P_g \ket{e} = \delta_{g,e}\ket{e}.
\end{equation}
Note that most elements act as \textit{zero} in the singlet representation, rather than as the identity. Still, this representation is rightfully thought of as being ``trivial'' in the sense that one has
\begin{equation}
    V \otimes \delta_e \cong \delta_e \otimes V \cong V.
\end{equation}
Finally, note that it follows from equation \eqref{eq:dual_grp_alg_basis_act} that the singlet projector acts on a tensor product of vectors labeled by generic group elements as
\begin{equation}
    P_e (\ket{g_1} \otimes \ket{g_2} \otimes \cdots \ket{g_n}) = \delta_{e,g_1g_2\cdots g_n} \ket{g_1} \otimes \ket{g_2} \otimes \cdots \ket{g_n}.
\end{equation}
Comparing with equation \eqref{eq:vertex_constraint_double}, this makes clear that the magnetic constraint at a vertex $V$ in fact acts as a singlet projection:
\begin{equation}
    \Pi_V = \rho_V(P_e).
\end{equation}
Therefore, like the electric constraint, the magnetic constraint is also a singlet condition, though now the role of the group is played by the dual group algebra.

While the basis for the dual group algebra used so far is useful for explicit computations, it obscures a duality between the electric and magnetic constraints of the model \cite{Buerschaper:2010yf}. An alternative choice of basis is provided by \textit{Wilson line operators}, which are defined as
\begin{equation}
    Z^a_{jk} = \sum_h \rho(h)^a_{jk}P_h,
\end{equation}
where $a$ labels an irreducible $G$-representation and $\rho(h)^a_{jk}$ is the matrix representing the group element $h$ with respect to some chosen orthonormal basis.\footnote{This is also referred to as the Fourier transformed basis in \cite{Buerschaper_2009}, as it is the basis provided by the Peter-Weyl theorem. This is reviewed in appendix \ref{app:peter-weyl}.} When defining the action of Wilson line operators on states, it is important to note that the dual group algebra possesses a separate \textit{left} and \textit{right} action on the edge Hilbert space. In the original operator basis these are defined to act on basis elements as (see \cite{Kitaev:double})
\begin{equation}
    L(P_g) \ket{h} := \delta_{g,h}\ket{h}, \quad R(P_g) \ket{h} := \delta_{g^{-1},h}\ket{h}.
\end{equation}
For Wilson line operators, borrowing notation from \cite{chung2025spontaneouslybrokennoninvertiblesymmetries}, we will distinguish between the left and right action with arrows, 
\begin{equation}
    (\rightZ{Z})_{jk}^a \ket{h} = \rho(h)^a_{jk}\ket{h}, \quad (\leftZ{Z})_{jk}^a \ket{h} = \rho(h)^{\overline{a}}_{kj}\ket{h},
\end{equation}
where $\overline{a}$ denotes the $G$-representation dual to $a$.

The term ``Wilson line'' for these operators is best understood via the gauge-theoretic language used to describe double models. If the edge variables of the double model are interpreted as transition functions of a principal $G$-bundle, then these operators act in the same way as Wilson line operators with a fixed choice of state at each endpoint. The geometry of this action is shown in figure \ref{fig:wilson_line_grp}.
\begin{figure}
    \centering
    \includegraphics[width=8.5cm]{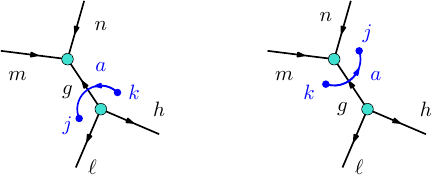}
    \caption{The geometry of how a Wilson line operator $Z_{jk}^a$ acts, on the left and right respectively, on an edge with a transition function $\ket{g}$.}
    \label{fig:wilson_line_grp}
\end{figure}

Closing a Wilson line involves tracing over its endpoint degrees of freedom, which produces a Wilson loop or \textit{character operator},
\begin{equation}
    \chi_a = \sum_{j} Z^a_{jj}, \quad \chi_a \ket{h} = \chi_a(h)\ket{h}.
\end{equation}
These operators have the important property that they satisfy the fusion algebra of $\Rep(G)$, which is the category of irreducible $G$-representations.
In other words, if the product of two irreducible representations takes the form
\begin{equation}
    V_a \otimes V_b = \bigoplus_{c} N_{ab}^c V_c,
\end{equation}
then the character operators satisfy
\begin{equation}
    \chi_a \chi_b = \sum_{c} N_{ab}^c \chi_c.
\end{equation}
It follows that there is a natural action of $\Rep(G)$ on the link Hilbert spaces, and that this action is embedded within the action of the dual group algebra. In fact, the singlet projector for $\mathbb{C}[G]^{\vee}$ can be written entirely in terms of this action as\footnote{This formula follows from the following two classic results in the representation theory of finite groups:
\begin{equation*}
    \frac{1}{|G|} \sum_a \text{dim}(V_a) \chi_a(h) = \delta_{h, e}, \quad |G| = \sum_a \dim(V_a)^2.
\end{equation*}
See for example \cite[chapter 2.2]{Fulton:book}.}
\begin{equation}\label{eq:repG_proj_formula}
    P_e = \frac{1}{D^2} \sum_a \dim(V_a) \chi_a, \quad D^2 := \sum_{a} \dim(V_a)^2.
\end{equation}

As noted below \eqref{eq:dual_rep_prod}, the definition of a tensor product of representations of $\C[G]^\vee$ requires a choice of ordering of the representations.\footnote{In section \ref{sec:projection-magnetic} we will specify this ordering at the vertex of a double model by defining $\C[G]^\vee$ to act on the data of a vertex and a distinguished edge attached to the vertex.} In appendix \ref{app:peter-weyl} we show that character operators act on tensor product representations as
\begin{equation}\label{eq:character_action_grp}
    \rightZ{\chi}_{a}(\ket{v_1} \otimes \ket{v_2} \cdots \otimes \ket{v_n}) = \sum_{j_1\cdots j_n} \left((\rightZ{Z})_{j_1 j_2}^a \ket{v_1} \otimes (\rightZ{Z})_{j_2 j_3}^a \ket{v_2} \otimes \cdots \otimes (\rightZ{Z})_{j_{n}j_1}^a\ket{v_n}\right),
\end{equation}
and therefore depend on an ordering only up to cyclic permutation. In order to define the action of the characters at a vertex, we therefore make use of the following two conventions:\footnote{The alternative choice of ingoing and outgoing convention is made in \cite{Kitaev:double}. This has the effect of reversing the orientations of Wilson line operators.}
\begin{enumerate}
    \item the tensor product in equation \eqref{eq:character_action_grp} is defined with a \textit{clockwise} labeling,
    \item an \textit{outgoing} edge receives the \textit{left} action and an \textit{ingoing} edge receives the \textit{right} action.
\end{enumerate}
With these choices, the singlet projector has a simple graphical representation when acting on a vertex as
\begin{equation}
    \includegraphics[width=5cm,valign=m]{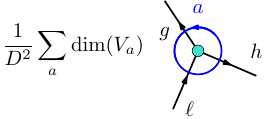}~.
\end{equation}
Note that this takes an analogous form to the electric constraint in the double model, which acts on a plaquette as
\begin{equation}
    \includegraphics[width=4.2cm,valign=m]{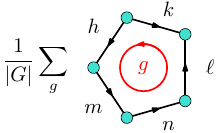}~.
\end{equation}

The rewriting of the magnetic constraint in terms of characters is conceptually important because it allows us to use the language of non-invertible symmetry. Initially, because the dual group algebra acts on the vertices of the model, one might be tempted to say that the double model has a local ``$\C[G]^\vee$ symmetry'' acting on vertices. But while it is technically correct that there is such an action, there are reasons for not using the term ``symmetry'' to describe this. Specifically, in local quantum systems it is often fruitful to reserve the word \textit{symmetry} for operations that preserve the locality structure of the theory.
In continuum quantum field theory, the modern point of view is that the most general algebraic structure that can appropriately be called a ``symmetry'' is that of a (higher) category.\footnote{For a recent discussion of how the intrinsic categorical symmetry of a theory is realized in practice in different backgrounds, see \cite{AliAhmad:2025bnd,Benini:2025lav}.}
From this point of view, $\Rep(G)$ has the right to be called a symmetry, while $\C[G]^{\vee}$ is a larger collection of operators associated with this intrinsic categorical symmetry.
Therefore, we will freely say that the double model has a local ``$\Rep(G)$ gauge symmetry'' acting on the vertices, and that the vertex constraint is the associated Gauss law.

Since a Gauss law is typically associated with the presence of gauge fields, one might want to further qualify why the term is appropriate here. To make this precise would require a formal definition of ``non-invertible gauge theory.'' While a general definition of such theories is far from being completely developed, examples of non-invertible gauge theories are provided by the Turaev-Viro theories \cite{TURAEV1992865, kirillov2010turaevviroinvariantsextendedtqft}.\footnote{See also \cite{Kawagoe:2024tgv} for a study of gauge-theoretic interpretations of string-net models.} These continuum field theories are structurally similar to discrete gauge theories, but replace group-valued transition functions with the data of a fusion category \cite{thorngren2019fusioncategorysymmetryi}. By considering an appropriate change of basis, the edge degrees of freedom of a double model can be made to resemble those of a Turaev-Viro theory together with permissible ``flux violations.''\footnote{This point is essentially developed in \cite{Buerschaper_2009}, combined with the known fact that the string-net Hilbert space with vertex constraint enforced produces the Turaev-Viro Hilbert space after enforcing a stabilizer condition \cite{kirillov2010turaevviroinvariantsextendedtqft}.}

Taken together, this means that the double model can be viewed as a lattice model of the $\Rep(G)$-Turaev-Viro theory that allows for flux violations, and hence as a kind of ``$\Rep(G)$-lattice gauge theory.'' We will take these points as sufficient justification for using the term ``Gauss law'' for the magnetic constraint. In section \ref{sec:hopf} we will apply the same logic to discuss more general classes of non-invertible Gauss laws. With this background and our terminology established, we now move on to an explicit study of Haag duality and additivity in magnetically constrained subspaces of the double model.

\section{A counterexample to exact Haag duality}
\label{sec:counterexample}
The point of this section is to show that when magnetic constraints are enforced exactly for a nonabelian group $G$, exact Haag duality need not be satisfied by the constrained algebras.

Given a region $R$, how do we determine the constrained algebra $\A_0(R)$ from equation \eqref{eq:constrained-algebras}? The first point is to note that for the magnetic constraint, the projection $\Pi_0$ is diagonal in the group basis.
Consequently, if we decompose a generic operator $O$ on the lattice as
\begin{equation}
    O = \sum_{\vec{g}, \vec{h}} c_{\vec{g}, \vec{h}} |g_1 \otimes \dots \otimes g_{|E|}\rangle \langle h_1 \otimes \dots \otimes h_{|E|}|,
\end{equation}
this operator is constraint-preserving if and only if each of the transition operators 
\begin{equation}
    |g_1 \otimes \dots \otimes g_{|E|}\rangle \langle h_1 \otimes \dots \otimes h_{|E|}|
\end{equation}
is constraint-preserving for $c_{\vec{g}, \vec{h}} \neq 0.$
Within a subregion $R$, the task of finding the algebra $\A_0(R)$ is the same as determining when the operator
\begin{equation} \label{eq:R-transition-operator}
    |g_1 \otimes \dots \otimes g_{|R|} \rangle \langle h_1 \otimes \dots \otimes h_{|R|} | \otimes 1_{R'}
\end{equation}
is constraint-preserving.

Consider a four-valent vertex that is part of a larger lattice; see figure \ref{fig:four-valent-vertex}.
Locally, the state of the system is labeled by a choice of group element at each edge.
Given the labeling in figure \ref{fig:four-valent-vertex}, the magnetic constraint for the vertex is $gh=k \ell.$
Suppose that we consider a region $R$ that contains two opposite edges of the vertex --- the $g$ and $\ell$ edges --- and does not contain the others --- the $h$ and $k$ edges.
An operator like the one in equation \eqref{eq:R-transition-operator} will act on these edges by sending $g$ to some other group element $g',$ and sending $\ell$ to $\ell',$ without changing $h$ or $k.$
Under what conditions can such a map be constraint-preserving?

\begin{figure}
    \centering
    \begin{tikzpicture}[
    scale=1,
    thick_line/.style={line width=1.5pt},
    directed/.style={
        postaction={decorate},
        decoration={
            markings,
            mark=at position 0.65 with {\arrow{Latex[length=3.5mm, width=2.5mm]}}
        }
    },
    label_node/.style={font=\Large}
    ]

    \draw[thick_line, directed, ForestGreen] (-2, 0) -- (0, 0) node[near start, above=4pt, label_node] {$g$};
    \draw[thick_line, directed, ForestGreen] (0, 0) -- (2, 0) node[near end, below=4pt, label_node] {$\ell$};
    \draw[thick_line, directed] (0, 0) -- (0, 2) node[near end, right =4pt, label_node] {$k$};
    \draw[thick_line, directed] (0, -2) -- (0, 0) node[near start, left=4pt, label_node] {$h$};

    \node [circle, draw=black, fill=cyan, minimum size=3pt] at (0,0) {};
    \end{tikzpicture}
    \caption{A four-valent vertex with edge labels $g, h, k,$ and $\ell.$ For the chosen orientations, the magnetic constraint is $g h = k \ell.$ The green-tinted ``$g$'' and ``$\ell$'' edges are in the region $R$.}
    \label{fig:four-valent-vertex}
\end{figure}
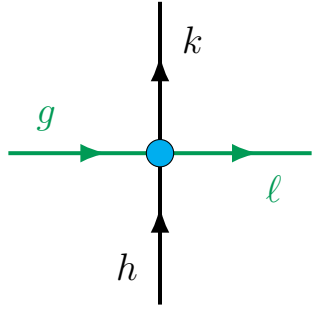

On quite general lattices, one can conclude that the only possible constraint-preserving maps are ones that multiply both $g$ and $\ell$ by a common element of the center of $G$.
This happens provided that the global structure of the lattice permits one to find, for any choice of $g, \ell,$ and $k,$ a global configuration that satisfies the constraint and that is compatible with these edges.\footnote{This places some global constraints on the lattice that prevent it from being ``too degenerate.'' For example, if the vertex shown in figure \ref{fig:four-valent-vertex} were a full lattice, with opposite edges identified periodically, then one would have the global constraints $g=\ell$ and $h=k.$}
Supposing this is possible, one finds that by choosing $k = e,$ there is a configuration satisfying the constraint,
\begin{equation} \label{eq:cusp-vertex}
    \includegraphics[width=4cm,valign=m]{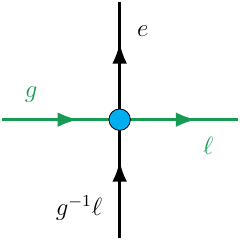}~,
\end{equation}
so one has
\begin{equation}
    h = g^{-1} \ell.
\end{equation}
For the map $g \mapsto g', \ell \mapsto \ell'$ to be constraint-preserving, one must have
\begin{equation}
    \ell' = g' h = g' g^{-1} \ell,
\end{equation}
or
\begin{equation} \label{eq:ell-g}
    \ell' \ell^{-1} = g' g^{-1}.
\end{equation}
Now letting $k$ be generic, one finds that there is a configuration satisfying the constraint with
\begin{equation}
    h
        = g^{-1} k \ell,
\end{equation}
and constraint preservation requires
\begin{equation}
    k \ell' = g' h = g' g^{-1} k \ell,
\end{equation}
or
\begin{equation}
    k (\ell' \ell^{-1}) k^{-1} = g' g^{-1}.
\end{equation}
Using equation \eqref{eq:ell-g}, this gives that $k$ conjugates $g' g^{-1}$ to itself, and since this must hold for all $k \in G,$ one concludes that $g' g^{-1}$ is in the center of $G$.

We conclude that for regions with ``opposite edges'' at a four-valent vertex, as in figure \ref{fig:four-valent-vertex}, mild assumptions on the global lattice structure restrict constraint-preserving operators to act via left-multiplication by a central element.
In particular, if $G$ is a nonabelian group with trivial center, then the only possible constraint-preserving map is $g \mapsto g, \ell \mapsto \ell.$
This setting is quite interesting, because it means that the projector $|g\rangle\langle g|$ is in both $\A_0(R)$ \textit{and} the commutant $\A_0(R)'.$

We will now use the above observations to construct a counterexample to exact Haag duality for a nonabelian group $G$ with trivial center.
To do this, we will construct a closed lattice with a region $R$ such that a projector $|g\rangle\langle g|$ is in $\A_0(R)'$, but not in $\A_0(R')$.
The lattice is drawn here, with top-bottom and left-right identification:
\begin{equation}
    \includegraphics[scale=0.8,valign=m]{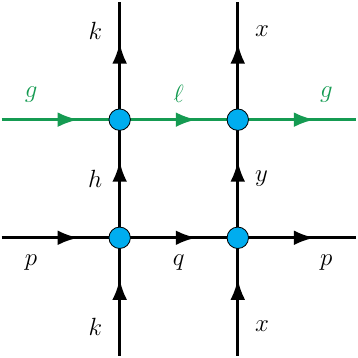}~.
\end{equation}
In the above labeling, it is true that for any choice of $g, \ell,$ and $k$, there is a global constraint-preserving configuration --- a particular example is sketched here:
\begin{equation}
    \includegraphics[scale=0.8,valign=m]{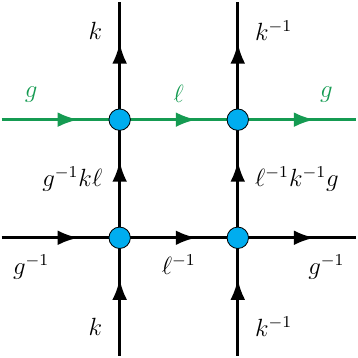}~.
\end{equation}
Consequently, by our above arguments, the projector $|g\rangle \langle g|$ is in both $\A_0(R)$ and $\A_0(R)'.$
We claim that this projector is not in $\A_0(R').$

Suppose, toward contradiction, that there were a constraint-preserving operator $O$ on $R'$ with the property that $O,$ restricted to the constraint subspace, acts as the projector $|g\rangle \langle g|.$
Then examine the two configurations satisfying the constraint here, where we assume $g' \neq g$:
\begin{equation}
    \includegraphics[scale=0.8,valign=m]{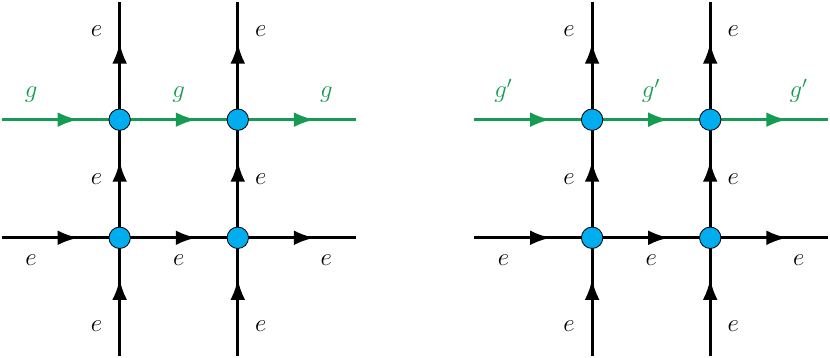}~.
\end{equation}
The operator $O$ would need to act as the identity on the left configuration, but it would need to annihilate the right configuration.
This is impossible, however, since the two configurations agree on the support of $O$.
We conclude that within the magnetically constrained algebra, there is a Haag duality violation:
\begin{equation}
    \A_0(R)' \supsetneq \A_0(R').
\end{equation}

\section{Non-abelian magnetic constraints and collars}
\label{sec:cuspless-magnetic}

\subsection{Cusps and weak Haag duality}
\label{sec:cusps}

The Haag duality violations of the preceding section occur in the presence of vertices that have ``cusps,'' like the one in equation \eqref{eq:cusp-vertex}.
More generally, we will say that a region $R$ has a cusp at vertex $V$ if the edges of $V$ contained in $R$ cannot be connected by a single adjacency path.
Equivalently, $R$ has a cusp at $V$ if the set of edges of $V$ contained in $R$ is not connected in the dual graph.
See figure \ref{fig:seven-valent-vertex} for an example of cusped and cuspless vertex on a many-valent lattice, and figure \ref{fig:continuum-picture} for intuition as to why we call these regions cusped.

\begin{figure}[h]
    \centering
    \includegraphics{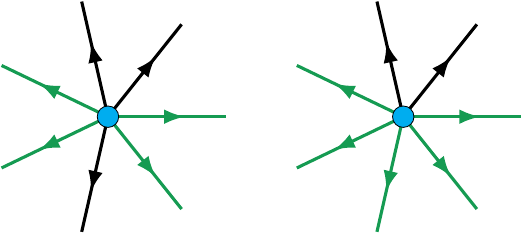}
    \caption{Two seven-valent vertices, with the edges of $R$ colored in green.
    The left vertex is cusped, while the right is not.}
    \label{fig:seven-valent-vertex}
\end{figure}

\begin{figure}[h]
    \centering
    \includegraphics{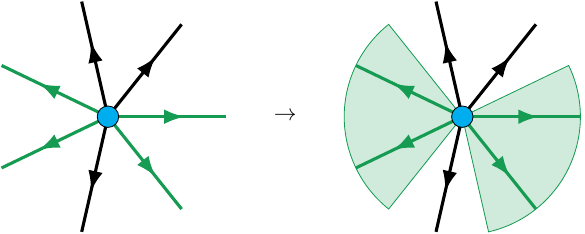}
    \caption{If one tries to make a ``continuum picture'' of a cusped vertex by drawing a continuous region that includes all edges of $R$ and excludes all edges of $R',$ this region must necessarily meet itself at a cusp, like in the right side of this figure.}
    \label{fig:continuum-picture}
\end{figure}

For the magnetic constraints associated to a $G$ double model, as explained in section \ref{sec:counterexample}, the set of constraint-preserving operators at a cusped vertex is somewhat delicate.
This can lead to Haag duality violations for regions containing cusped vertices.
To improve the situation, we define a ``collared'' region $R^+$ by starting with $R$, then adding all edges that belong to cusped vertices of $R$; see for example figure \ref{fig:collar}.
We will demonstrate the ``weak Haag duality'' relation\footnote{The same notion appears in a different but closely related context in \cite{Evans:fusion}.}
\begin{equation} \label{eq:collared-HD}
    \A_0(R^+)' \subseteq \A_0(R').
\end{equation}
When $R$ has no cusped vertices, this implies the inclusion $\A_0(R)' \subseteq \A_0(R').$
The reverse ($\supseteq$) inclusion is obvious from the definitions in section \ref{sec:background}, so equation \eqref{eq:collared-HD} implies exact Haag duality for cuspless regions.\footnote{Note that a trivalent lattice has no cusped vertices, so Haag duality holds exactly under magnetic constraints on any trivalent lattice.}

\begin{figure}[h]
    \centering
    \includegraphics{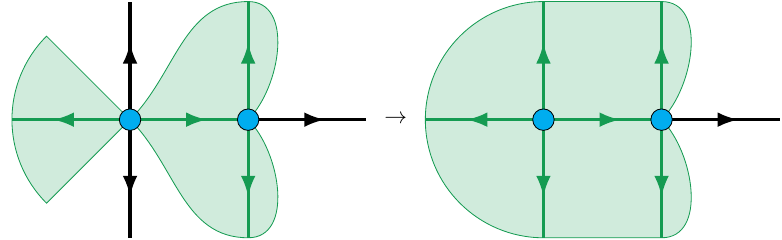}
    \caption{Exchanging a region $R$ for its collared version $R^+$. Only one vertex of $R$ is cusped, and this is the only one that gets new edges in $R^+.$}
    \label{fig:collar}
\end{figure}

\subsection{Projection formulas and lifts}\label{subsec:grp_lifts}

Before proceeding, it will be informative to revisit the setting of \cite{Harlow:disjoint-additivity}, in which exact Haag duality was found to hold under the imposition of all \textit{electric} constraints.
The relevant subspace $\H_0$ of the double-model Hilbert space is the one stabilized by all plaquette projection operators
\begin{equation}
    \Pi_P = \frac{1}{|G|}\sum_{g \in G} \rho_P(g),
\end{equation}
where $\rho_P(g)$ is the electric gauge transformation corresponding to the group element $g.$
In this setting, Haag duality holds exactly:
\begin{equation} \label{eq:exact-HD}
    \A_0(R)' = \A_0(R').
\end{equation}

As discussed in section \ref{subsec:qdouble}, the local operator algebra $\cA_0(R) \subseteq \mathcal{B}(\cH_0)$ is defined as the subspace of operators having a lift (within the unconstrained space) to an operator with support on the region $R$.
So to study the set $\A_0(R)',$ one must understand when an operator $\omega \in \A_0(R)'$ can be lifted to an operator $\tilde{\omega}$ with a given support.
In particular, to prove equation \eqref{eq:exact-HD}, one must construct, for each $\omega \in \A_0(R)',$ a lift $\tilde{\omega}$ supported on $R'.$

The key tool used in constructing lifts is a version of an algebraic notion called \textit{compression duality}.
More details can be found in \cite[section 3.2]{Harlow:disjoint-additivity}.
Let $\omega$ be an operator acting on $\H_0,$ and suppose that there exists a lift $\tilde{\omega}$ supported within the region $R'$.
Since the lift $\tilde{\omega}$ commutes with all operators in $R,$ it must act on the space $\A(R) \H_0$ as
\begin{equation} \label{eq:lift-definition}
    \tilde{\omega} a |\psi\rangle
        = a \omega |\psi\rangle, \quad |\psi\rangle \in \H_0, \quad a \in \A(R).
\end{equation}
Starting with an operator $\omega$ acting on $\H_0$, one might therefore try to \textit{define} a lift to act on this subspace via equation \eqref{eq:lift-definition}, and define $\tilde{\omega}$ to vanish on the orthocomplement of this subspace. This will provide a sort of ``coarsest'' lift of $\omega$, with its support being all of $R'$.
For this definition to be unambiguous and so well defined, one must have the implication
\begin{equation}
    \left( \sum_{j=1}^{n} a_j |\psi_j\rangle = 0 \right)
        \quad \Rightarrow \quad \left( \sum_{j=1}^{n} a_j \omega |\psi_j\rangle = 0 \right), \quad \ket{\psi_j} \in \cH_0, \quad a_j \in \cA(R).
\end{equation}
One can compute
\begin{equation}
    \left\lVert \sum_{j=1}^{n} a_j \omega |\psi_j\rangle \right\rVert^2
        = \sum_{j,k=1}^{n} \langle \psi_j | \omega^{\dagger} a_j^{\dagger} a_k \omega |\psi_j \rangle
        = \sum_{j,k=1}^n \langle \psi_j | \omega^{\dagger} (\Pi_0 a_j^{\dagger} a_k \Pi_0) \omega |\psi_j \rangle,
\end{equation}
where $\Pi_0$ is the projector onto $\H_0.$
So if $\Pi_0 a_j^{\dagger} a_k \Pi_0$ commutes with $\omega,$ then whenever $\omega$ is an isometry, the operator $\tilde{\omega}$ is an isometry as well, and in particular $\tilde{\omega}$ is well defined.
So when $\omega$ is an isometry, this gives
\begin{equation}
    \omega \in (\Pi_0 \A(R) \Pi_0)'
        \quad \Rightarrow \quad
        \tilde{\omega} \text{ exists in } \A(R'),
\end{equation}
and since any von Neumann algebra is generated by its isometries, one finds the inclusion
\begin{equation} \label{eq:compression-inclusion}
    (\Pi_0 \A(R) \Pi_0)'
        \subseteq \A_0(R').
\end{equation}
In particular, to prove exact Haag duality in the presence of an electric constraint, it sufficed in \cite{Harlow:disjoint-additivity} to prove
\begin{align} \label{eq:electric-projection-formula}
        \A_0(R) = \Pi_0 \A(R) \Pi_0
   \end{align}
for all regions $R$.

While equation \eqref{eq:electric-projection-formula} will \textit{not} hold in the presence of a magnetic constraint (or a more general non-invertible Gauss law), we will eventually show the inclusion
\begin{equation}
    \text{magnetic constraint:}
    \quad
    \Pi_0 \A(R) \Pi_0 \subseteq \A_0(R^+),
\end{equation}
which will prove the claim \eqref{eq:collared-HD} via the logic
\begin{equation}
    \A_0(R^+)' \subseteq (\Pi_0 \A(R) \Pi_0)' \subseteq \A_0(R').
\end{equation}
To prepare for that argument, we first give a simple, representation-theoretic way of understanding why equation \eqref{eq:electric-projection-formula} holds for the electric Gauss law constraint in finite-group gauge theory.

Given a unitary representation of a finite group $\mathcal{S}$ on a finite-dimensional Hilbert space $\H,$ one obtains an adjoint representation of $\mathcal{S}$ on the operators of $\H$:
\begin{equation}
    \mathcal{O} \in \B(\H) \quad \to \quad s \cdot \mathcal{O} = U(s) \mathcal{O} U(s)^{\dagger}.
\end{equation}
With respect to the Hilbert-Schmidt inner product on operators,
\begin{equation}
    \llangle \mathcal{O}_1 | \mathcal{O}_2 \rrangle = \tr(\mathcal{O}_1^{\dagger} \mathcal{O}_2),
\end{equation}
this representation is unitary, and any $\mathcal{O} \in \B(\H)$ can be decomposed, under this action, into irreducible representations of $\S$ as
\begin{equation}
    \mathcal{O} = \sum_{a} \mathcal{O}_{a}.
\end{equation}
The component of $\mathcal{O}$ corresponding to the trivial representation,
\begin{equation}
    \mathcal{O} = \mathcal{O}_{0} + \dots,
\end{equation}
will be called the \textit{neutral} piece of $\mathcal{O}$.\footnote{Generally we will use the term ``neutral'' for an operator that lives in the trivial representation of the adjoint action, and ``singlet'' for a state that lives in the trivial representation of the original action of $\S$ on $\H$.} There are two important facts about $\mathcal{O}_0.$
The first is that $\mathcal{O}_0$ preserves the singlet sector of $\H$.
This is because for any $s \in \S,$ one has
\begin{equation}
    \mathcal{O}_0 U(s) = U(s) \mathcal{O}_0,
\end{equation}
so if $U(s)$ preserves $|\psi\rangle,$ then it also preserves $\mathcal{O}_0 |\psi\rangle.$ The second fact is that, when restricted to acting on the singlet, $\mathcal{O}_0$ acts in the same way as $\mathcal{O}$:
\begin{equation}
    \Pi_0 \mathcal{O}_0 \Pi_0 = \Pi_0 \mathcal{O} \Pi_0.
\end{equation}
This can be shown by direct calculation, using that $\mathcal{O}_0$ can be computed from the projection formula,
\begin{equation}\label{eq:grp_op_singlet_projector}
    \mathcal{O}_0 = \frac{1}{|S|}\sum_{s \in S} U(s) \mathcal{O} U(s)^\dagger.
\end{equation}
Acting on a singlet state, one therefore finds
\begin{equation}
    \mathcal{O}_0 \ket{\psi} = \frac{1}{|S|}\sum_{s \in S} U(s) \mathcal{O} U(s)^\dagger\ket{\psi} = \frac{1}{|S|}\sum_{s \in S} U(s) \mathcal{O} \ket{\psi} = \Pi_0 \mathcal{O} \ket{\psi},
\end{equation}
as desired.

On a closed lattice with $N_P$ plaquettes, the $G$-double model carries an action of $G$ at every plaquette, and the electric constraint is imposed by going to the singlet sector of the group $\S = G^{\times N_P}.$ The neutral piece of a generic operator $\mathcal{O}$ is given by \eqref{eq:grp_op_singlet_projector}, and because each $U(s)$ is a tensor product of unitary operators, the support of $\mathcal{O}_0$ is no larger than the support of $\mathcal{O}.$ As a result, whenever $\mathcal{O}$ is in $\A(R),$ we have
\begin{equation}
    \Pi_{0} \mathcal{O} \Pi_0
        = \Pi_0 \mathcal{O}_0 \Pi_0 \in \A_0(R),
\end{equation}
which gives
\begin{equation}
    \Pi_0 \A(R) \Pi_0
        \subseteq \A_0(R).
\end{equation}
The other inclusion in \eqref{eq:electric-projection-formula} is obvious, so this completes the proof of equation \eqref{eq:electric-projection-formula}.
Note that the operator $\Pi_0 \mathcal{O} \Pi_0$ itself, when interpreted as an operator acting on $\H$, will generally have support outside of $R$; the important point is that $\Pi_0 \mathcal{O} \Pi_0$ is in $\A_0(R),$ since there exists a \textit{distinct} lift $\mathcal{O}_0$ that lives within $\A(R).$

In the case of magnetic constraints, we will have a similar representation-theoretic argument for the projection formula $\Pi_0 \A(R) \Pi_0 \subseteq \A(R^+).$
The main difference will be that the projection will be onto the singlet sector of a local $\text{Rep}(G)$ action on Hilbert space, not a local $G$ action; consequently, the projection can mildly increase the support of an operator. 

\subsection{Projection formulas and magnetic constraints}
\label{sec:projection-magnetic}

Now we will shift back to the case of the magnetic constraint in a double model based on a finite group $G$.
We will prove that Haag duality holds exactly for any cuspless region $R$ and demonstrate a weaker condition on cusped regions.
As discussed above, if $\Pi_0$ is the projector onto the constraint subspace, we will reach these conclusions by proving
\begin{align}\label{eq:compression_gen}
    \Pi_0 \A(R) \Pi_0 \subseteq \A_0(R^+).
\end{align}

In the case of electric constraints, it was useful to think of the electric constraint as being implemented by a singlet projection in a representation of the group $G$.
Fortunately, as explained in section \ref{subsec:gauge_theory_rep}, the magnetic constraint can be thought of as a singlet projection in a representation of the dual group algebra $\mathbb{C}[G]^{\vee}.$
Recall that this algebra has a basis $P_g$ labeled by $g \in G,$ and acts on an edge of the double model by
\begin{equation}
    P_g |h\rangle = \delta_{g, h} |h\rangle.
\end{equation}
Given a vertex $V$ and a choice of adjacent edge $E$,\footnote{This plays the same role as the choice of \textit{site} in \cite{Kitaev:double}.} one can order the edges of $V$ clockwise starting from $E$, and define an action of $\mathbb{C}[G]^{\vee}$ on $(V, E)$ via the tensor product representation
\begin{equation} \label{eq:double-vertex-action}
    \rho^{(V, E)}(P_h) (|g_1\rangle \otimes \dots \otimes |g_n\rangle)
        \equiv \delta_{h, g_1 \dots g_n} (|g_1\rangle \otimes \dots \otimes |g_n\rangle).
\end{equation}
The magnetic constraint at this vertex is then imposed by acting with the element $P_e$, which --- as explained in section \ref{subsec:gauge_theory_rep} --- has the representational interpretation of projecting onto a singlet of $\C[G]^\vee$.

In the same way that a unitary representation of a group incurs an adjoint action on operators, any representation $\rho$ of $\C[G]^\vee$ on a Hilbert space $\cH$, satisfying $\rho(P_g)^{\dagger} = \rho(P_g),$\footnote{For a dual group algebra, this is the correct compatibility condition of a representation with the inner product. See appendix \ref{app:hopf_def}.} induces an adjoint representation
\begin{equation}
    \cO \in \B(\H) \quad \to \quad P_g \cdot \cO = \sum_{h} \rho(P_{h}) \cO \rho(P_{g^{-1}h}).
\end{equation}
As in the case of a group action, this decomposes an operator into $\C[G]^\vee$-irreducible representations,
\begin{equation}
    \cO = \sum_a \cO_a,
\end{equation}
with the neutral component satisfying the two properties
\begin{enumerate}
    \item $\cO_0$ is singlet-preserving, and
    \item on the singlet subspace, $\cO_0$ is equal to $\Pi_0 \cO \Pi_0$.
\end{enumerate}
Both properties follow immediately from the explicit projection formula
\begin{equation}
    \cO_0 = \sum_{g} \rho(P_{g}) \cO \rho(P_{g}),
\end{equation}
which is the analogue of \eqref{eq:grp_op_singlet_projector} in the present setting.\footnote{The origin of this formula is explained further in the context of general Hopf algebras in appendix \ref{app:hopf}.} Acting on a singlet state $|\psi\rangle$, we have
\begin{equation}
    \cO_0 \ket{\psi} = \sum_{g} \rho(P_{g}) \cO \rho(P_{g}) \ket{\psi} = \sum_{g} \rho(P_{g}) \delta_{g,1} \cO \ket{\psi} = \rho(P_{e}) \cO \ket{\psi},
\end{equation}
which is both another singlet state and equal to the action of $\Pi_0 \cO \Pi_0$ on the same state, since we have $\rho(P_e) = \Pi_0.$
This means that, from a representational point of view, we are in the same setting as in section \ref{subsec:grp_lifts}. 

Applied to the quantum double model, the neutral piece of an operator $\cO$ at a vertex is then given by
\begin{equation}
    \cO^{(V,E)}_0
        = \sum_{g \in G} \rho^{(V, E)}(P_g) \cO \rho^{(V, E)}(P_g).
\end{equation}
Let us now consider the support of the operator $\cO_0^{(V,E)}$ as compared to the support of the operator $\cO.$
Obviously, the support of $\cO_0^{(V,E)}$ can only modify the support of $\cO$ by adding edges that are adjacent to $V$, but we can be more precise than this. 
Here we will work with concrete expressions to help orient the reader, leaving a more abstract but general discussion to section \ref{sec:hopf}. 

Momentarily letting $\cO$ be supported only on the edges of $V$, so we can write
\begin{equation}
    \cO = \sum_{g_1, \dots, g_n; h_1, \dots, h_n} c_{g_1, \dots, g_n; h_1, \dots, h_n} |g_1 \otimes \dots \otimes g_n\rangle \langle h_1 \otimes \dots \otimes h_n|,
\end{equation}
we have
\begin{align}
    \begin{split}
    \cO^{(V,E)}_0
        = \sum_{x \in G} &
        \left(\sum_{y_1 \dots y_n = x} |y_1 \otimes \dots \otimes y_n \rangle \langle y_1 \otimes \dots \otimes y_n|\right) \\
        & \times \left( \sum_{g_1, \dots, g_n; h_1, \dots, h_n} c_{g_1, \dots, g_n; h_1, \dots, h_n} |g_1 \otimes \dots \otimes g_n\rangle \langle h_1 \otimes \dots \otimes h_n| \right) \\
        & \times \left(\sum_{z_1 \dots z_n = x} |z_1 \otimes \dots \otimes z_n \rangle \langle z_1 \otimes \dots \otimes z_n|\right),
    \end{split}
\end{align}
which we can simplify as
\begin{align}
    \begin{split}
    \cO^{(V,E)}_0
        = \sum_{x \in G} \sum_{g_1 \dots g_n = x} \sum_{h_1 \dots h_n = x}  c_{g_1, \dots, g_n; h_1, \dots, h_n} |g_1 \otimes \dots \otimes g_n\rangle \langle h_1 \otimes \dots \otimes h_n|,
    \end{split}
\end{align}
or simplify even further as
\begin{align}
    \begin{split}
    \cO^{(V,E)}_0
        = \sum_{g_1 \dots g_n = h_1 \dots h_n }  c_{g_1, \dots, g_n; h_1, \dots, h_n} |g_1 \otimes \dots \otimes g_n\rangle \langle h_1 \otimes \dots \otimes h_n|,
    \end{split}
\end{align}
Now suppose that $\cO$ is not supported on all edges of $V$; instead it is only supported between some edges $j$ and $k$ with respect to the ordering induced by $(V,E).$
This is equivalent to writing
\begin{equation}
    c_{g_1, \dots, g_n; h_1, \dots, h_n}
    = \delta_{g_1, h_1} \dots \delta_{g_{j-1}, h_{j-1}} c_{g_j, \dots, g_k; h_j, \dots, h_k} \delta_{g_{k+1}, h_{k+1}} \dots \delta_{g_{n}, h_{n}}
\end{equation}
Then we have that $\cO^{(V,E)}_0$ is 
\begin{equation}
    = \text{id}_{1} \otimes \dots \otimes \text{id}_{j-1} \otimes \left( \sum_{h_j \dots h_k = g_j \dots g_k}  c_{g_j, \dots, g_k; h_j, \dots, h_k}|g_1 \otimes \dots \otimes g_n\rangle \langle  h_j \otimes \dots \otimes h_k | \right) \otimes \text{id}_{k+1} \otimes \dots \otimes \text{id}_{n}.
\end{equation}
So the support of $\cO_0^{(V,E)}$ can grow only ``in between'' edges $j$ and $k$ with respect to the ordering induced by $(V,E).$
We sketch this for a four-valent vertex in figure \ref{fig:site-projection}.

\begin{figure}
    \centering
    \includegraphics{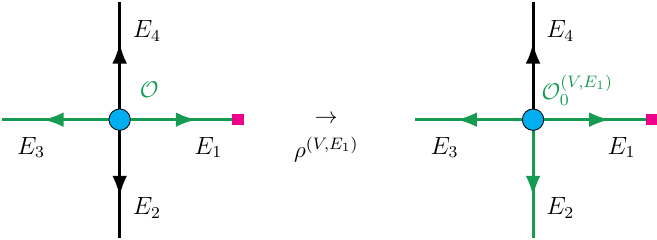}
    \caption{A neutral projection at a vertex $V$ with the adjacent edge $E_1$ chosen to label the action of the dual group algebra.
    When $\cO$ is originally supported on edges $E_1$ and $E_3,$ the singlet projection $\rho^{(V,E_1)}$ creates an operator $\cO_0^{(V,E_1)}$ that has support on $E_1, E_2,$ and $E_3$.}
    \label{fig:site-projection}
\end{figure}

What do we learn from this? 
For every edge $V$ in the region $R$ that is cuspless, it is possible to choose a pairing $(V,E)$ such that when $\cO$ is supported in $R$, the operator $\cO_0^{(V,E)}$ is supported in $R$ as well.
For every cusped vertex, any choice of adjacent edge will cause the support of $\cO_0^{(V,E)}$ to be larger than that of $V$.
But if we choose an edge at each vertex such that \textit{the support of $\cO$ does not grow at cuspless vertices}, then perform the simultaneous singlet projection of all $\rho^{(V,E)}$ actions, we obtain a neutral operator $\cO_0$ for which the support is contained in $R^+.$
We therefore have
\begin{equation}
    \Pi_0 \cO \Pi_0
        = \Pi_0 \cO_0 \Pi_0 \subseteq \A_0(R^+),
\end{equation}
which proves
\begin{equation}
    \Pi_0 \A(R) \Pi_0
        \subseteq \A_0(R^+),
\end{equation}
which we already anticipated in section \ref{subsec:grp_lifts}.
There, we demonstrated that this inclusion implies the weak Haag duality relation
\begin{equation}\label{eq:weak_haag}
    \A_0(R^+)'
        \subseteq \A_0(R'),
\end{equation}
which gives exact Haag duality for cuspless regions.

\subsection{Disjoint additivity}\label{subsec:disjoint_add_group}

So far we have discussed Haag duality in the presence of magnetic constraints, and its sensitivity to discreteness artifacts in the form of ``weak Haag duality'' or ``Haag duality with a collar.''
But Haag duality is not the only locality principle studied in relation with operator algebras; as discussed in section \ref{sec:background-algebraic}, one often also studies the principle of additivity. Even in the continuum, additivity is not always satisfied for generic regions, due to the possibility of having topologically nontrivial operators that cannot be broken into contractible pieces.
To get around this, \cite{Harlow:disjoint-additivity} introduced the notion of ``disjoint additivity,'' which essentially means that additivity should hold for regions that are appropriately ``disjoint.''

We will now investigate this notion for the local algebras $\A_0(R)$ associated with the magnetic constraint of a double model. Take $R_1$ and $R_2$ to be two regions. The inclusions
\begin{equation}
    \cA_0(R_1) \subseteq \cA_0(R_1 \cup R_2), \quad \cA_0(R_2) \subseteq \cA_0(R_1 \cup R_2),
\end{equation}
induce another inclusion
\begin{equation}
    \cA_0(R_1) \vee \cA_0(R_2) \subseteq \cA_0(R_1 \cup R_2).
\end{equation}
The opposite inclusion may or may not be satisfied depending on the details of $R_1$ and $R_2.$
If we take an element
\begin{equation}
    \omega \in \cA_0(R_1 \cup R_2),
\end{equation}
and let $\w{\omega}$ be a constraint preserving lift, then additivity of the tensor product algebra implies that that we can write this operator as
\begin{equation}
    \w{\omega} = \sum_{i} \w{\omega}^{(1)}_i \w{\omega}^{(2)}_i, \quad  \w{\omega}^{(j)}_i \in \cA(R_j).
\end{equation}
While the total sum is constraint-preserving, for a generic such decomposition the operators $\w{\omega}^{(j)}_i$ need \textit{not} be individually constraint preserving. Generalizing the approach of \cite{Harlow:disjoint-additivity} for the electric constraint, one can attempt to remedy this by taking neutral components. To do this, we will need the following proposition.
\begin{proposition} \label{prop:neutral-DA-projection}
    If $R_1$ and $R_2$ are disjoint in the sense that there does not exist a vertex that has some edges contained in $R_1$ and some contained in $R_2$, then we have
    \begin{equation}
        \left(\w{\omega}^{(1)}_i \w{\omega}^{(2)}_i\right)_0 = \left(\w{\omega}^{(1)}_i\right)_0 \left(\w{\omega}^{(2)}_i\right)_0.
    \end{equation}
\end{proposition}
\begin{proof}
    By definition, 
    \begin{equation}
        \left(\w{\omega}^{(1)}_i \w{\omega}^{(2)}_i\right)_0 = \prod_V \sum_g \rho(P_g)^{(V,E)} \w{\omega}^{(1)}_i \w{\omega}^{(2)}_i \rho(P_g)^{(V,E)},
    \end{equation}
    for some arbitrary choice of ordering at each vertex. By the disjointness condition, we can decompose this product according to whether a vertex touches $R_1$, $R_2$, or neither. For a vertex meeting neither region, we have
    \begin{equation}
        \sum_g \rho(P_g)^{(V,E)}\rho(P_g)^{(V,E)} = \text{id}_V.
    \end{equation}
    The product therefore splits as
    \begin{equation}
        = \left(\prod_{V \in R_1} \sum_g \rho(P_g)^{(V,E)} \w{\omega}^{(1)}_i \rho(P_g)^{(V,E)}\right)  \left(\prod_{V \in R_2} \sum_g \rho(P_g)^{(V,E)} \w{\omega}^{(2)}_i \rho(P_g)^{(V,E)}\right),
    \end{equation}
    which is by definition the product of each neutral component,
    \begin{equation}
        = \left(\w{\omega}^{(1)}_i\right)_0 \left(\w{\omega}^{(2)}_i\right)_0.
    \end{equation}
    Note that in this argument, it was crucial that the operators were disjoint in the ``vertex-adjacent'' sense of the proposition statement, otherwise the product would not have factorized. 
\end{proof}
\noindent Since the operators $\w{\omega}$ and $\w{\omega}_0$ agree on the constraint subspace, this gives us a new lift of $\omega$, now decomposed into a product of constraint preserving operators,
\begin{equation}
    \w{\omega}_0 = \sum_{i} \left(\w{\omega}^{(1)}_i\right)_0 \left(\w{\omega}^{(2)}_i\right)_0, \quad \left(\w{\omega}^{(j)}_i\right)_0 \in \cA(R_j^+),
\end{equation}
where we have emphasized that the neutral component of an operator may have its support extended into a collared region. Altogether, this argument produces an inclusion
\begin{equation}\label{eq:weak_add_cond}
    \cA_0(R_1 \cup R_2) \subseteq \cA_0({R_1^+}) \vee \cA_0({R_2^+}).
\end{equation}
This gives the following corollary.
\begin{corollary} \label{cor:DA-cuspless}
    If $R_1$ and $R_2$ are both cuspless and disjoint as above, then additivity holds,
    \begin{equation}
        \cA_0(R_1 \cup R_2) = \cA_0({R_1}) \vee \cA_0({R_2}).
    \end{equation}
    In particular, this disjoint form of additivity holds on any trivalent lattice.
\end{corollary}

For regions with cusped vertices, this method of argument provides a weaker bound than that needed to prove additivity. It is therefore natural to ask if exact disjoint additivity can be violated for non-cuspless regions. In fact, we will now show that this is \textit{not} true --- for the magnetic constraint of a $G$ double model, disjoint additivity holds even for regions with cusps.

To do this, we will take advantage of a useful basis of operators. For a region $R$, we define the operators
\begin{equation}
    \cO(\vec{g},\vec{h}) = |\vec{g}\rangle \langle\vec{h}| = \prod_{E \in R} \ket{g_E}\bra{h_E},
\end{equation}
where $\vec{g}$ and $\vec{h}$ are vectors labeling group elements at each edge of $R$.
The operator $\cO(\vec{g}, \vec{h})$ is defined to act as the identity outside of $R$.
It is clear that these operators span $\A(R).$
We may therefore express a general constraint-preserving operator in $\A(R_1 \cup R_2)$ as
\begin{equation} \label{eq:additivity-lift}
    \w{\omega} = \sum_{\vec{g}, \vec{h}, \vec{k}, \vec{\ell}} C(\vec{g}, \vec{h}, \vec{k}, \vec{\ell})\cO(\vec{g},\vec{h})^{(1)}\cO(\vec{k},\vec{\ell})^{(2)},
\end{equation}
with $\cO(\vec{g},\vec{h})^{(i)}$ an element of $\cA(R_i)$ as above. The condition that $\w{\omega}$ is constraint-preserving can be written
\begin{equation}
    \w{\omega} \Pi_0 = \Pi_0 \w{\omega} \Pi_0 \Leftrightarrow (1 - \Pi_0)\w{\omega} \Pi_0 = 0.
\end{equation}
In our explicit decomposition of $\w{\omega},$ this condition is
\begin{equation} \label{eq:group-O-decomposition}
    \sum_{\vec{g}, \vec{h}, \vec{k}, \vec{\ell}} C(\vec{g}, \vec{h}, \vec{k}, \vec{\ell}) (1 - \Pi_0)\cO(\vec{g},\vec{h})^{(1)}\cO(\vec{k},\vec{\ell})^{(2)}\Pi_0 = 0.
\end{equation}
We will now argue that the terms in this sum are all mutually orthogonal with respect to the Hilbert-Schmidt inner product on $\H$.

To see this, we study the explicit Hilbert-Schmidt inner product of two terms in the sum:
\begin{equation}
    \tr\left(\Pi_0 \cO(\vec{g},\vec{h})^{(1)\dagger} \cO(\vec{k},\vec{\ell})^{(2) \dagger} (1 - \Pi_0)(1 - \Pi_0)\cO(\vec{p},\vec{q})^{(1)}\cO(\vec{r},\vec{s})^{(2)}\Pi_0\right).
\end{equation}
To simplify this expression, we use cyclicity of the trace, together with $\Pi_0^2 = \Pi_0,$ to rewrite the inner product as
\begin{equation}
    \tr\left(\cO(\vec{g},\vec{h})^{(1)\dagger} \cO(\vec{k},\vec{\ell})^{(2) \dagger} (1 - \Pi_0)\cO(\vec{p},\vec{q})^{(1)}\cO(\vec{r},\vec{s})^{(2)}\Pi_0\right).
\end{equation}
Next, we use the fact that the projection operator $\Pi_0$ acts diagonally in the group basis.
If one inserts an identity on the complement of $R_1 \cup R_2$ of the form $\sum_{\vec{x}} |\vec{x}\rangle\langle \vec{x}|^{(c)},$ then the inner product we are trying to compute is
\begin{equation}
    \sum_{\vec{x}} \tr\left(\cO(\vec{g},\vec{h})^{(1)\dagger} \cO(\vec{k},\vec{\ell})^{(2) \dagger} (1 - \Pi_0)\cO(\vec{p},\vec{q})^{(1)}\cO(\vec{r},\vec{s})^{(2)} |\vec{x}\rangle\langle \vec{x}|^{(c)} \Pi_0\right),
\end{equation}
and diagonality of $\Pi_0$ means that the total operator being traced over is proportional to the same operator with the factors of $(1-\Pi_0)$ and $\Pi_0$ removed:
\begin{align}
\begin{split}
    & \tr\left(\cO(\vec{g},\vec{h})^{(1)\dagger} \cO(\vec{k},\vec{\ell})^{(2) \dagger} (1 - \Pi_0)\cO(\vec{p},\vec{q})^{(1)}\cO(\vec{r},\vec{s})^{(2)} |\vec{x}\rangle\langle \vec{x}|^{(c)} \Pi_0\right) \\
        & \qquad \propto \tr\left(\cO(\vec{g},\vec{h})^{(1)\dagger} \cO(\vec{k},\vec{\ell})^{(2) \dagger} \cO(\vec{p},\vec{q})^{(1)}\cO(\vec{r},\vec{s})^{(2)} |\vec{x}\rangle\langle \vec{x}|^{(c)} \right).
\end{split}
\end{align}
Since we have $\cO(\vec{g}, \vec{h})^{(1) \dagger} = \cO(\vec{h}, \vec{g})^{(1)},$ and we clearly have $\mathcal{O}(\vec{h}, \vec{g})^{(1)} \mathcal{O}(\vec{p}, \vec{q})^{(1)} = \delta_{\vec{g}, \vec{p}} \mathcal{O}(\vec{h}, \vec{q})^{(1)},$ this expression simplifies to
\begin{equation}
    \delta_{\vec{g}, \vec{p}} \delta_{\vec{\ell}, \vec{r}} \tr\left(\cO(\vec{h},\vec{q})^{(1)} \cO(\vec{k}, \vec{s})^{(2)} |\vec{x}\rangle\langle \vec{x}|^{(c)} \right)
        \propto \delta_{\vec{g}, \vec{p}} \delta_{\vec{\ell}, \vec{r}} \delta_{\vec{h}, \vec{q}} \delta_{\vec{k}, \vec{s}}.
\end{equation}
Combining all of these calculations, we have established the claim that the individual terms in equation \eqref{eq:group-O-decomposition} are mutually Hilbert-Schmidt orthogonal.

Since the constraint-preserving condition is that the full sum in equation \eqref{eq:group-O-decomposition} vanishes, we may conclude that for a constraint-preserving operator, each individual term in the sum must vanish; therefore, we have
\begin{equation}
    C(\vec{g},\vec{h},\vec{k},\vec{\ell})(1 - \Pi_0)\cO(\vec{g},\vec{h})^{(1)}\cO(\vec{k},\vec{\ell})^{(2)}\Pi_0 = 0.
\end{equation}
That is, either the coefficient vanishes or the basis element itself is constraint-preserving.
In the latter case, it is clear that when acting on the constrained subspace, the operator $\cO(\vec{g},\vec{h})^{(i)}$ can only produce a state violating the constraint at vertices that contain an edge in $R_i$. Therefore, if there is no vertex that has support in both $R_1$ and $R_2$, then the operator can be constraint preserving if and only if each $\cO(\vec{g},\vec{h})^{(i)}$ is constraint preserving individually.
The lift of $\omega \in \A_0(R_1 \cup R_2)$ provided by equation \eqref{eq:additivity-lift} is therefore seen explicitly to be expressible as a sum of constraint-preserving operators in $R_1$ multiplied by constraint-preserving operators in $R_2$; in other words, we have $\hat{\omega} \in \A(R_1 \cup R_2)$ and $\hat{\omega}$ constraint-preserving, hence $\omega \in \A_0(R_1) \vee \A_0(R_2).$
This gives the inclusion $\A_0(R_1 \cup R_2) \subseteq \A_0(R_1) \vee \A_0(R_2)$, which gives disjoint additivity:
\begin{proposition}
    For the magnetic constraint of a $G$ double model, if $R_1$ and $R_2$ are disjoint in the sense of vertex adjacency, then
    \begin{equation}
        \cA_0(R_1 \cup R_2) = \cA_0(R_1) \vee \cA_0(R_2).
    \end{equation}
\end{proposition}

\section{Generalizing to Hopf algebras} \label{sec:hopf}
So far, we have studied operator algebras in the magnetically constrained subspace of a quantum double model based on a finite group $G$. In our analyses of both Haag duality and additivity, a key role was played by the representation theory of the dual group algebra $\C[G]^\vee$. The representation theory of the dual group algebra finds its natural home in the theory of \textit{finite-dimensional $C^*$-Hopf algebras}. These are algebraic structures that generalize the structure of a finite group, being closely related to so-called finite quantum groups  \cite{Drinfeld:1986in}. They have the characterizing property that their representation theory behaves like that of a finite group: their representations have tensor products, one dimensional singlet representations, and charge conjugates. For the reader unfamiliar with Hopf algebras, we provide a review of basic definitions and structure theory in appendix \ref{app:hopf}.

It was noted already in the original paper \cite{Kitaev:double} that quantum double models admit an immediate generalization where the finite group is replaced by a finite-dimensional C$^*$-Hopf algebra.\footnote{As all of our Hopf algebras will be finite-dimensional and C$^*$; the term Hopf algebra will always mean such Hopf algebras unless otherwise noted.} Such models are based on a tensor product over edges of a lattice, with the Hilbert space of a single edge chosen to be a finite-dimensional C$^*$-Hopf algebra, $\H_e = H$, with the inner product determined in a natural way by the structure of $H$.
Vertex (magnetic) and plaquette (electric) constraints are then defined using representation-theoretic constructions and are used to construct a commuting projector Hamiltonian.
These constructions are reviewed explicitly in section \ref{subsec:hopf_double}. 

Hopf algebra double models were first explicitly constructed in \cite{Buerschaper_2013, Buerschaper:2010yf}. Such models live in the same universality class as certain string-net models, and the close relationship to both string-net models and Turaev-Viro theories was developed in \cite{Buerschaper:2010yf, balsam2012kitaevslatticemodelturaevviro}. As explained in \cite{Buerschaper:2010yf}, Hopf algebra double models behave as ``extended string-nets,'' i.e., as string-net models with additional degrees of freedom at the vertices. Much in the same way as in a double model built on a group, the electric and magnetic constraints of a Hopf algebra double model can be thought of as Gauss laws for non-invertible gauge symmetries. Hopf algebra double models might therefore be thought of as a kind of non-invertible lattice gauge theory. This point is implicit already in \cite{Buerschaper:2010yf}.

In this section we generalize our analysis of algebraic locality to the setting of a Hopf algebra double model.
Specifically, we prove that under the magnetic constraint of such a model, weak Haag duality is satisfied on general lattices, with exact Haag duality satisfied for cuspless regions.
Similarly, we prove that the models satisfy a weak form of disjoint additivity on general lattices, and exact disjoint additivity for cuspless regions.\footnote{A duality between the double model based on $H$ and the double model based on the dual $H^{\vee}$, discussed in \cite{Buerschaper:2010yf}, implies that all of the same results hold if the electric constraint is imposed in place of the magnetic one, provided that ``cusps'' and ``disjointness'' are defined relative to the dual graph.} 

\subsection{Hopf algebra double models}\label{subsec:hopf_double}
In this section we review the construction of a Hopf algebra double model \cite{Buerschaper_2013}, following more closely the construction in \cite{Buerschaper:2010yf}. We provide definitions and necessary background on Hopf algebras in appendix \ref{app:hopf}. As remarked above, the term Hopf algebra will always be taken to mean finite-dimensional C$^*$-Hopf algebra.

A Hopf algebra quantum double model is defined on an oriented two-dimensional lattice embedded in a closed, oriented $2$-manifold. The edge Hilbert spaces are associated to a finite-dimensional C$^*$-Hopf algebra,
\begin{equation}
    \cH_e = H,
\end{equation}
and the total Hilbert space is the tensor product over edges.
Each edge Hilbert space carries a natural left and right action by $H$,
\begin{equation}
    L_\alpha \ket{\beta} = \ket{\alpha \beta}, \quad R_\alpha \ket{\beta} = \ket{\beta S(\alpha)}, \quad \alpha,\beta \in H,
\end{equation}
and a natural left and right action by $H^\vee$,
\begin{equation}
    L_\lambda \ket{\alpha} = \lambda(\alpha^{(2)})|\alpha^{(1)}\rangle, \quad R_\lambda \ket{\alpha} = S(\lambda)(\alpha^{(1)})|\alpha^{(2)}\rangle, \quad \lambda \in H^\vee.
\end{equation}
In the above, $H^\vee$ denotes the dual Hopf algebra. In defining the $H^{\vee}$-action we have employed Sweedler notation, which is a useful convention reviewed in appendix \ref{app:hopf_def}. Note that we have also defined the ``right actions'' with an antipode map so as to obtain a left action in the labeling variable.
This is standard \cite{Kitaev:double}, and is convenient when defining vertex and plaquette constraints.

In general there is not a canonical basis for $H$. Instead, $H$ has two natural decompositions.
One decomposition is
\begin{equation}\label{eq:hopf_decomp}
    H \cong \bigoplus_a \End(V_a) \cong \sum_a V_a^\vee \otimes V_a,
\end{equation}
where the sum is over irreducible representations of the dual Hopf algebra $H^\vee$.
The other decomposition is
\begin{equation}\label{eq:hopf_decomp_wedderburn}
    H \cong \bigoplus_{x} \End(W_{x}) \cong \sum_{x} W_{x}^\vee \otimes W_{x}, 
\end{equation}
where now the sum is over irreducible representations of the Hopf algebra $H$ itself.
Both decompositions are representation-theoretic in nature. The first follows from the natural left and right action of $H^\vee$ on $H$, and is called the Peter-Weyl decomposition.
The second follows from the left and right action of $H$ on itself and is called the Wedderburn-Artin decomposition. 

For the familiar case of a double model constructed from a group algebra $\C[G]$, the representations of $\C[G]^{\vee}$ are one-dimensional and labeled by $g \in G$, so the decomposition of equation \eqref{eq:hopf_decomp} is just the decomposition of $\C[G]$ into the group basis.
This is a special case in which $H$ does have a canonical basis.
In general, one must choose an orthonormal basis on each irreducible representation of $H^{\vee},$ and from this one obtains a basis on $H$ expressed as 
\begin{equation}
    Z_{ij}^a \in V_a^{\vee} \otimes V_a \subseteq H,
\end{equation}
with are uniquely characterized by the property\footnote{In this formula, $\rho(\lambda)^a_{ij}$ are the matrix elements of the representative of $\lambda$ in $V_a$. The fact that this formula fully characterizes $Z_{ij}^{a}$ follows from the canonical isomorphism $(H^\vee)^\vee \cong H$.}
\begin{equation}
    \lambda(Z_{ij}^a) = \rho(\lambda)_{ij}^a, \quad \lambda \in H^\vee.
\end{equation}
In this basis, the left and right actions of $H^\vee$ take a particularly simple form\footnote{These expressions can equally well be written as
\begin{equation*}
    L_\lambda \ket{Z_{ij}^a} = \rho(S(\lambda))_{jk}^{\overline{a}}\ket{Z_{ik}^a}, \quad R_\lambda \ket{Z_{ij}^a} = \rho(S(\lambda))_{ik}^{a}\ket{Z_{kj}^a}.
\end{equation*}
}
\begin{equation}
    L_\lambda \ket{Z_{ij}^a} = \rho(\lambda)_{kj}^a\ket{Z_{ik}^a}, \quad R_\lambda \ket{Z_{ij}^a} = \rho(\lambda)_{ki}^{\bar{a}}\ket{Z_{kj}^a}.
\end{equation}
We review this fact in appendix \ref{app:peter-weyl}. The subscripts therefore behave as basis elements in the given irreducible representation and its dual. Finally, in this basis, the natural inner product is defined as
\begin{equation}
    \langle Z_{jk}^a | Z_{\ell m}^b\rangle := \frac{\delta^{ab}}{d_a} \delta_{j\ell}\delta_{k m}.
\end{equation}
We review the representation-theoretic construction of this inner product in appendix \ref{app:haar}.

The basic structure we have just reviewed has a useful graphical presentation, in which a basis element $Z^a_{jk}$ is drawn as a line labeled by the representation, with two nodes labeled by elements of the representation and its dual:
\begin{equation} \label{eq:hopf-decomp}
    \includegraphics[width=4.8cm,valign=m]{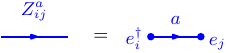}~.
\end{equation}
Here $e_j$ is an orthonormal basis element of $V_a$ and $e^\dagger_i$ is a dual basis element of $V_a^\vee$. For example, a model on a hexagonal lattice can be graphically presented as
\begin{equation}
    \includegraphics[width=2.4cm,valign=m]{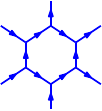} \quad = \quad \includegraphics[width=4.4cm,valign=m]{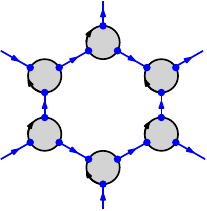}~,
\end{equation}
where the left figure corresponds to a labeling of edges as in the left side of equation \eqref{eq:hopf-decomp}, and the right figure corresponds to a labeling as in the right side of equation \eqref{eq:hopf-decomp}.
In other words, on the right, we have ``blown up'' the vertices into disks so as to allow every edge to carry a basis degree of freedom at each of its incident vertices \cite{Buerschaper_2009}. 
For future convenience, the disks in this figure are given an orientation that is determined by the orientation of the two-manifold, and by convention we will choose the orientation of the disks to be clockwise. 

Having discussed the state space of the model in detail, we may now proceed to the electric and magnetic constraints.
These constraints are related to the \textit{Haar integral} and \textit{Haar measure} on $H$, which are most concretely defined in terms of character elements within the Hopf algebra and its dual:
\begin{equation}
    \chi_a = \sum_{j}Z^a_{jj} \in H, \quad \chi_x = \sum_{j}Z^x_{jj} \in H^\vee.
\end{equation}
The characters have the special property that they furnish representations of the fusion rings of $\Rep(H^\vee)$ and $\Rep(H)$ respectively:
\begin{equation}
    \chi_a\chi_b = \bigoplus_c N_{ab}^c \chi_c, \quad \chi_x\chi_y = \bigoplus_z M_{xy}^z \chi_z.
\end{equation}
For this reason, by analogy to the realization of a fusion category by defect operators in a continuum theory, we will sometimes refer to the character elements as defects within the Hopf algebras.
In the spirit of the end of section \ref{subsec:gauge_theory_rep}, this means that an $H$-action is associated to a $\Rep(H^\vee)$ symmetry, and similarly for the $H^{\vee}$ action and a $\Rep(H)$ symmetry.

In terms of the defects, the Haar measure and Haar integral are defined as the elements
\begin{equation}\label{eq:haar_int}
    h = \frac{1}{D^2}\sum_{a} \dim(V_a) \chi_a, \quad h \in H,
\end{equation}
and
\begin{equation}\label{eq:haar_measure}
    \mu = \frac{1}{D^2} \sum_{x} \dim(W_x) \chi_{x}, \quad \mu \in H^\vee,
\end{equation}
where to simplify notation we have defined\footnote{These equalities follow from the two decompositions \eqref{eq:hopf_decomp} and \eqref{eq:hopf_decomp_wedderburn}.}
\begin{equation}
    D^2 = |H| = \sum_a \dim(V_a)^2 = \sum_{x} \dim(W_x)^2.
\end{equation}
Both the Haar integral and Haar measure have a simple graphical action on plaquettes and vertices as \cite{Buerschaper_2009}\footnote{These figures are given here primarily so that the familiar reader can see the close connection to the plaquette operator in a string-net model \cite{Lin_2021, Buerschaper_2009,Buerschaper:2010yf}.}
\begin{equation}\label{eq:int_actions}
    \frac{1}{D^2}\sum_a \dim(V_a) ~\includegraphics[width=4cm,valign=m]{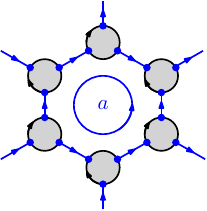}, \quad\; \frac{1}{D^2}\sum_{x} \dim(W_x) ~\includegraphics[width=4.3cm,valign=m]{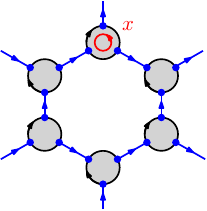}.
\end{equation}

There are a few necessary remarks to make about the Haar integral and Haar measure. First, they are \textit{singlet projectors}. In a Hopf algebra, the singlet representation is defined in terms of a distinguished algebra homomorphism called the \textit{counit},
\begin{equation}
    \epsilon: H \to \C, \qquad {\epsilon}_\vee: H^\vee \to \C,
\end{equation}
and the singlet is defined as the one-dimensional representation satisfying
\begin{equation}
    \alpha\ket{\Omega} = \epsilon(\alpha)\ket{\Omega}.
\end{equation}
Both the Haar integral and Haar measure have the properties
\begin{equation}
    h^2 = h, \quad \mu^2 = \mu,
\end{equation}
and
\begin{equation}
    \alpha h = \epsilon(\alpha) h, \quad \lambda \mu = {\epsilon}_\vee(\lambda) \mu, \quad \alpha \in H, \; \lambda \in H^\vee.
\end{equation}
This last property can be understood in many ways, including via representation theory or by using techniques of topological field theory. We provide one such discussion in appendix \ref{app:haar}. The important point is that, in the language of section \ref{subsec:gauge_theory_rep}, the plaquette constraint functions as a \textit{Gauss law} (electric constraint) for a local $\Rep(H^\vee)$ symmetry, while the vertex constraint functions as a flux condition (magnetic constraint) that is enforced by a local $\Rep(H)$ symmetry.
These roles are swapped in an electric-magnetic duality that is provided by exchanging the lattice with its dual and exchanging $H$ with $H^\vee$ \cite{Buerschaper:2010yf}.

Above, we defined the action of a character at a vertex or plaquette.
To define the action of a general element of $H$ or $H^{\vee}$, one must make a choice of ordering at vertices and plaquettes.
This is because the action on a tensor product of edges is defined by the coproduct,
\begin{equation}
    \lambda \cdot (\ket{\alpha_1} \otimes \cdots \otimes \ket{\alpha_n}) =  \lambda^{(1)}\ket{\alpha_1} \otimes \cdots \otimes \lambda^{(n)}\ket{\alpha_n}, \quad \lambda \in H^\vee,
\end{equation}
and in general the iterated coproduct of an element is \textit{not} permutation-invariant. By fixing a choice of orientation on the spatial surface, an ordering of edges around a vertex or plaquette is equivalent to specifying an initial choice of edge; the rest of the ordering is determined by going clockwise around vertices or counterclockwise around plaquettes.

Based on the above considerations, we may denote the action of $H^{\vee}$ on a vertex by\footnote{As previously remarked, this choice is equivalent to a choice of site in \cite{Kitaev:double}.} $\rho^{(V,E)}(\lambda)$; similar notation can be used for the action of $H$ on plaquettes, though we will not make use of this below.
Note that the singlet projectors of both $H^{\vee}$ on vertices and $H$ on plaquettes are independent of the choice of edge; this is because the singlet projections can be written in terms of the characters as above.
Another way of seeing this is that $\mu$ and $h$ are \textit{cocommutative}, meaning that they satisfy
\begin{equation}\label{eq:cyclic_haar_measure}
    \mu^{(\sigma(1))} \otimes \cdots \otimes \mu^{(\sigma(n))} = \mu^{(1)} \otimes \cdots \otimes \mu^{(n)}, \quad \sigma \in S_n,
\end{equation}
for any cyclic permutation $\sigma$.\footnote{This is shown to be equivalent to a more standard definition of cocommutative in appendix \ref{app:hopf_def}.}

There remains one final ambiguity in defining the action of $H^{\vee}$ or $H$ at a vertex or a plaquette, which is a convention for left versus right actions.
At a vertex, we will follow the same convention as in section \ref{subsec:gauge_theory_rep}: outgoing edges receive a left action and ingoing edges receive a right action. At a plaquette, we will follow a convention specified in section \ref{subsec:qdouble}: counterclockwise-oriented edges receive the left action, while clockwise-oriented edges receive the right action.\footnote{Note that this is the opposite convention as compared to \cite{Kitaev:double}. The conventions are exchanged by acting with antipodes of elements.} With this choice of convention, defects are naturally oriented \textit{counterclockwise}, as in equation \eqref{eq:int_actions}. With all of this said, the Hamiltonian of the model is built from the commuting projectors \cite{Buerschaper:2010yf, Buerschaper_2013}
\begin{equation}
    A_V := \rho_{V}(\mu), \quad B_P := \rho_P(h).
\end{equation}

This review of the model shows that, from a representation-theoretic point of view, we find ourselves in the same setting as section \ref{sec:cuspless-magnetic}. By placing our prior arguments in their proper representation-theoretic home, we now generalize our earlier results to this larger class of models.

\subsection{Non-invertible Gauss laws and locality}\label{subsec:general_locality} 
Having reviewed the construction of Hopf algebra double models, we now move to studying local operator algebras in the magnetically constrained subspace. The total magnetic constraint operator is defined as the product of the commuting projectors,
\begin{equation}
    \Pi_0 = \prod_V \rho_V(\mu).
\end{equation}
As just explained, this operator projects each vertex onto its singlet subspace under the local $H^\vee$-action. Conceptually, following the discussion in section \ref{subsec:gauge_theory_rep}, this functions as the Gauss law constraint for a $\Rep(H)$ non-invertible gauge symmetry. 

The constrained algebras that we will study are defined as
\begin{equation} \tag{\ref{eq:constrained-algebras}}
    \A_{0}(R)
        \equiv \{ \Pi_{0} a \Pi_{0} \text{ such that } a \in \A(R) \text{ and } a \Pi_{0} = \Pi_{0} a \Pi_{0}\},
\end{equation}
where $\cA(R)$ is a local algebra in the unconstrained, tensor product Hilbert space. As mentioned in section \ref{sec:background-algebraic}, $\A_0(R)$ can be thought of as a subspace of the operator algebra on the constrained subspace,
\begin{equation}
    \cA_0(R) \subseteq \B(\cH_0), \quad \Pi_0(\cH) = \cH_0.
\end{equation}
Under this inclusion, $\cA_0(R)$ is the collection of operators acting on $\H_0$ that have a \textit{lift} to a constraint-preserving operator, acting on the unconstrained Hilbert space, with support in $R$. 

As in section \ref{sec:cuspless-magnetic}, the key step in studying both Haag duality and additivity for this local operator algebra will be proving the relation
\begin{equation}\tag{\ref{eq:compression_gen}}
    \Pi_0 \cA(R) \Pi_0 \subseteq \cA_0(R^+),
\end{equation}
where the ``collared'' region $R^+$ was defined in section \ref{sec:cusps}.
To do so, we will essentially repeat the analysis of section \ref{sec:cuspless-magnetic} in more general, Hopf-algebraic terms.
The most important point in that analysis was that any local operator may be decomposed into irreducible representations of a local $\Rep(G)$ symmetry, with explicit formulas for the neutral component of this decomposition.
The same basic structure will be seen to hold in the general, Hopf-algebraic setting.

To explain this point, consider first a generic Hilbert space $\cH$ carrying a unitary representation\footnote{A unitary representation is a representation in which the adjoint satisfies a certain compatibility condition with the Hopf algebra, as reviewed in appendix \ref{app:hopf_def}.} of a finite-dimensional C$^*$-Hopf algebra $H$. Operators acting on $\H$ inherit a natural conjugation action of $H$ via the formula\footnote{By comparing to \eqref{eq:grp_hopf_data}, one sees that this indeed correctly reduces to the conjugation action in the case of a group algebra.}
\begin{equation}
    \alpha \cdot \cO = \alpha^{(1)}\cO S(\alpha^{(2)}), \quad \cO \in \B(\cH), \; \alpha \in H.
\end{equation}
Because the operator algebra is an $H$-representation, it therefore decomposes into irreducible representations of $H$\footnote{This direct sum decomposition follows because a finite-dimensional C$^*$ algebra is semisimple.}
\begin{equation}
    \B(\cH) = \bigoplus_x (\B(\cH))_x,
\end{equation}
and so individual operators decompose as
\begin{equation}
    \cO = \sum_x \cO_x,
\end{equation}
where $x$ labels irreducible representations of $H$. We will refer to the component of the operator corresponding to the singlet representation as its neutral component. By definition, the Hopf algebra acts on a neutral operator as
\begin{equation}
    \alpha \cdot \cO_0 = \epsilon(\alpha) \cO_0.
\end{equation}
The neutral component of an operator has a concrete expression in terms of the Haar integral (singlet projector),
\begin{equation} \label{eq:neutral-comp-abstract}
    \cO_0 = h^{(1)} \cO S(h^{(2)}).
\end{equation}
This formula for $\mathcal{O}_0$, though abstract, is sufficient for our purposes. A more concrete form, making closer contact to the more familiar projection formula of a group is can be obtained by using the concrete form of the Haar integral \eqref{eq:haar_int}: 
\begin{equation}\label{eq:neutral_comp_explicit}
    \cO_0 = \frac{1}{|H|}\sum_{a} d_a Z_{ij}^a \cO (Z^{a}_{ij})^*,
\end{equation}
with this sum now over irreducible representations of $H^\vee$. This formula, which is in clear analogy to the neutral projection for a group, is derived in appendix \ref{app:haar}.

The neutral component of an operator satisfies the important property
\begin{equation} \label{eq:hopf-neutral-neutral}
    \cO_0 h = h \cO h.
\end{equation}
This is the same condition that we demonstrated for the special case of a group algebra and dual group algebra in section \ref{sec:cuspless-magnetic}. This property follows from direct computation, using the projection formula involving the Haar integral:
\begin{equation}
    \cO_0 h = h^{(1)} \cO S(h^{(2)})h = h^{(1)}\epsilon(S(h^{(2)})) \cO h = h^{(1)}\epsilon(h^{(2)}) \cO h = h\cO h.
\end{equation}
In the above computation, we have used several identities concerning antipodes and counits that are reviewed in appendix \ref{app:hopf_def}. Equations \eqref{eq:neutral-comp-abstract} and \eqref{eq:hopf-neutral-neutral} will be our primary tools going forward.

We now return to the specific case of a double model built on a Hopf algebra $H$. Applying the above discussion to the action of $H^{\vee}$ at each vertex, the neutral component of an operator $\cO$ under a vertex action takes the form\footnote{Since it is the \textit{dual} Hopf algebra that acts on vertices in the double model, the Haar measure is the correct singlet projector to use.}
\begin{equation}
    \cO_0^{(V,E)} = \rho^{(V,E)}(\mu^{(1)}) \cO \rho^{(V,E)}(S(\mu^{(2)})).
\end{equation}
Applying such neutral projections at each vertex, we obtain an operator that agrees with $\Pi_0 \cO \Pi_0$ on the constraint subspace. We would now like to study its support, which we will do in two steps.

First. At a vertex having $n$ edges, the operator $\cO^{(V,E)}_0$ has the explicit form
\begin{equation}
    \cO^{(V,E)}_0 := \left(\mu^{(1)} \otimes \cdots \otimes \mu^{(n)}\right) \cO \left(S(\mu^{(2n)}) \otimes \cdots \otimes S(\mu^{(n+1)})\right).
\end{equation}
Let us now suppose that the operator $\cO$ has no support on the final edge in this ordering. Then the component $\mu^{(n)}$ can be commuted through $\cO$ and the expression simplified to
\begin{align}
    &= \left(\mu^{(1)} \otimes \cdots \otimes \mu^{(n-1)} \otimes 1\right) \cO \left(S(\mu^{(2n)}) \otimes \cdots \otimes S(\mu^{(n+2)})\otimes \mu^{(n)}S(\mu^{(n+1)})\right), \\
    &= \left(\mu^{(1)} \otimes \cdots \otimes \mu^{(n-1)} \otimes 1\right) \cO \left(S(\mu^{(2n)}) \otimes \cdots \otimes S(\mu^{(n+2)})\otimes \epsilon(\mu^{(n)})\right), \\
    \intertext{which has used the identity $\mu^{(n)}S(\mu^{(n+1)}) = \epsilon(\mu^{(n)})$,}
    &= \left(\mu^{(1)} \otimes \cdots \otimes \mu^{(n-1)}\epsilon(\mu^{(n)}) \otimes 1\right) \cO \left(S(\mu^{(2n)}) \otimes \cdots \otimes S(\mu^{(n+2)})\otimes 1\right), \\
    &= \left(\mu^{(1)} \otimes \cdots \otimes \mu^{(n-1)} \otimes 1\right) \cO \left(S(\mu^{(2(n-1))}) \otimes \cdots \otimes S(\mu^{(n)})\otimes 1\right),
\end{align}
which used the identity $\mu^{(n-1)}\epsilon(\mu^{(n)}) = \mu^{(n-1)}$ \textit{and} reindexed the total coproduct. We therefore have the following lemma.
\begin{lemma}
    If $\cO$ does not have support on the final edge, then neither does $\cO_0^{(V,E)}$.
\end{lemma}

\noindent Second. A similar, but slightly different argument holds for the initial edge. First it will be convenient to rewrite the neutral operator as 
\begin{equation}
    \cO^{(V,E)}_0 := \left(\mu^{(2n)} \otimes \mu^{(1)} \otimes \cdots \otimes \mu^{(n-1)}\right) \cO \left(S(\mu^{(2n-1)}) \otimes \cdots \otimes S(\mu^{(n)})\right),
\end{equation}
using the fact that coproducts of $\mu$ are unchanged under cyclic permutations. Then, proceeding in the same way, we have,
\begin{align}
    &= \left(\mu^{(2n)}S(\mu^{(2n-1)}) \otimes \mu^{(1)} \otimes \cdots \otimes \mu^{(n-1)}\right) \cO \left(1 \otimes S(\mu^{(2n-2)}) \otimes \cdots \otimes S(\mu^{(n)})\right), \\
    &= \left(\epsilon(\mu^{(2n-1)}) \otimes \mu^{(1)} \otimes \cdots \otimes \mu^{(n-1)}\right) \cO \left(1 \otimes S(\mu^{(2n-2)}) \otimes \cdots \otimes S(\mu^{(n)})\right), \\
    \intertext{which has used the identity $\mu^{(2n)}S(\mu^{(2n-1)}) = \epsilon(\mu^{(2n-1)})$,\footnotemark}\noalign{\footnotetext{This identity follows the antipode identity and the fact that in a finite-dimensional C$^*$-Hopf algebra, the antipode is its own inverse.}}
    &= \left(1 \otimes \mu^{(1)} \otimes \cdots \otimes \mu^{(n-1)}\right) \cO \left(1 \otimes \epsilon(\mu^{(2n-1)})S(\mu^{(2n-2)}) \otimes \cdots \otimes S(\mu^{(n)})\right), \\
    &= \left(1 \otimes \mu^{(1)} \otimes \cdots \otimes \mu^{(n-1)}\right) \cO \left(1 \otimes S(\mu^{(2n-2)}) \otimes \cdots \otimes S(\mu^{(n)})\right),
\end{align}
which used the identity $S(\mu^{(2n-2)}\epsilon(\mu^{(2n-1)})) = S(\mu^{(2n-2)})$. Therefore, we have:
\begin{lemma}
    If $\cO$ does not have support on the first edge, then neither does $\cO_0^{(V,E)}$.
\end{lemma}
\noindent Finally, iterating both arguments, we have the complete result:
\begin{proposition}
    The support of $\cO_0^{(V,E)}$ does not extend beyond the outer edges of the support of $\cO$ with respect to the specified ordering.
\end{proposition}
\noindent Importantly, the support of $\cO_0^{(V,E)}$ \textit{can} expand within those edges. This spread in support will generically happen unless the dual Hopf algebra is cocommutative.\footnote{A cocommutative finite dimensional C$^*$-Hopf algebra is always the group algebra of some group $G$ \cite{milnor1965structure} and so this would reduce the studying the standard Gauss law of a $G$ symmetry \cite{Shao:additivity}.} This is the same behavior depicted in figure \ref{fig:site-projection}. 

As the behavior of the spreads of support is identical to that of section \ref{sec:projection-magnetic}, our earlier results follow immediately. In particular, defining a collared region as in that section, the inclusion
\begin{equation}\tag{\ref{eq:compression_gen}}
    \Pi_0 \cA(R) \Pi_0 \subseteq \cA_0(R^+),
\end{equation}
follows, and so a weak form of Haag duality holds for general regions:
\begin{theorem}
    For a general region $R$, the constrained subalgebra satisfies \textit{weak} Haag duality
    \begin{equation}\tag{\ref{eq:weak_haag}}
    \A_0(R^+)'
        \subseteq \A_0(R').
\end{equation}
\end{theorem}

It follows again that on cuspless regions, Haag duality is exact.
Consequently, Haag duality holds for all regions when the double model is studied on a trivalent lattice. Moreover, the explicit exact Haag duality violation demonstrated in section \ref{sec:counterexample} shows that the non-collared equality need not hold. We have not, however, demonstrated that the non-collared identity indeed fails for all non-grouplike Hopf algebra double models.

\subsection{Disjoint additivity}\label{subsec:dis_add}
Having completed our study of Haag duality, we now turn to disjoint additivity.
The results at the beginning of section \ref{subsec:disjoint_add_group} --- i.e., those appearing up through corollary \ref{cor:DA-cuspless} --- will apply verbatim in the case of a general Hopf algebra.
Specifically, we will momentarily show that a Hopf double model obeys the following proposition.
\begin{proposition} \label{prop:DA-hopf}
    If $R_1$ and $R_2$ are such that no vertex has an edge contained in both $R_1$ and $R_2$, then there is an inclusion
    \begin{equation} \label{eq:collared-DA}
        \cA_0(R_1 \cup R_2) \subseteq \cA_0({R_1^+}) \vee \cA_0({R_2^+}).
    \end{equation}
    In particular, if the regions are cuspless, then additivity holds,
    \begin{equation}
        \cA_0(R_1 \cup R_2) = \cA_0({R_1}) \vee \cA_0({R_2}).
    \end{equation}
\end{proposition}

Unlike in the case of the dual group algebra, we do not make any claims about disjoint additivity holding exactly for cusped regions; however, the weak form of disjoint additivity in equation \eqref{eq:collared-DA} does hold, and in particular this means that disjoint additivity holds exactly for any Hopf double model on a trivalent lattice.
Note that, as mentioned in the introduction, we have \textit{not} produced a counterexample to the stronger form of disjoint additivity; our suspicion is that disjoint additivity holds even for cusped regions, so long as the two regions do not share a vertex.

We now prove proposition \ref{prop:DA-hopf}.
\begin{proof}
    To extend the proof of proposition \ref{prop:neutral-DA-projection}, we need only show that given two operators 
    \begin{equation}
        \cO^{(j)} \in \cA(R_j), \quad j = 1,2,
    \end{equation}
    then
    \begin{equation}
        \left(\cO^{(1)} \cO^{(2)}\right)_0 = \left(\cO^{(1)}\right)_0\left(\cO^{(2)}\right)_0.
    \end{equation}
    The rest of the proof from  section \ref{subsec:disjoint_add_group} then goes through immediately.
    
    For a Hopf double, the neutral operator takes the form
    \begin{equation}
        \left(\cO^{(1)} \cO^{(2)}\right)_0 =  \prod_V \mu_V^{(1)} \cO^{(1)}\cO^{(2)} S(\mu_V^{(2)}).
    \end{equation}
    Now, if we assume that $R_1$ and $R_2$ have the property that there exists no vertex that has edges intersecting both $R_1$ and $R_2$, then we can decompose this product as,
    \begin{equation}
        = \prod_{V \in R_1} \prod_{V \in R_2} \prod_{V \not\in R_1, R_2} \mu_V^{(1)} \cO^{(1)}\cO^{(2)} S(\mu_V^{(2)}),
    \end{equation}
    and commute operators past each other to obtain,
    \begin{equation}
        =  \left(\prod_{V \in R_1}  \mu_V^{(1)} \cO^{(1)}  S(\mu_V^{(2)})\right) \left(\prod_{V \in R_2}  \mu_V^{(1)}\cO^{(2)} S(\mu_V^{(2)})\right) = \sum_i \left(\cO^{(1)}\right)_0 \left(\cO^{(2)}\right)_0.
    \end{equation}
    Therefore, we have the desired result.
\end{proof}

\section{Future directions}
\label{sec:discussion}

In this paper, we studied the algebraic locality structure of certain lattice constraints associated to non-invertible symmetries.
Concretely, we examined general quantum double models based on a finite-dimensional C$^*$ Hopf algebra, and explained how to think of the magnetic constraint of such a model as a non-invertible Gauss law.
We showed that when this constraint is imposed exactly, the local operator algebras satisfy weak forms of Haag duality and disjoint additivity.
These ``weak'' forms of algebraic locality become exact for nice, ``cuspless'' regions, and in particular, they hold for arbitrary regions on a trivalent lattice.
In the special case of an ordinary grouplike double model, we showed a ``strong'' form of disjoint additivity; the question of whether this strong form of disjoint additivity holds in the general Hopf-algebraic setting remains open.

There are several natural directions that would be interesting to explore in future work, mostly having to do with a larger class of models based on noninvertible symmetries.
One natural candidate is to study models based on \textit{weak} Hopf algebras, which were proposed in \cite{Buerschaper:2010yf} and constructed explicitly in \cite{Chang_2014}.
These models are more complicated, but they have the nice feature that they can realize any unitary fusion category as a local symmetry; by contrast, Hopf algebra double models can only realize non-anomalous symmetries.

Another natural direction is to use our techniques to study the locality structure of local algebras in a Hopf algebra double model when \textit{both} electric and magnetic constraints are imposed, since our current analysis only treats the case where one imposes one constraint or the other.
Since the only nontrivial operators in the full double model ground space are understood to be the ribbon operators of \cite{Kitaev:double}, it would be interesting to see if our techniques provide a novel construction of the ribbon algebra.

\acknowledgments{We thank Wilbur Shirley and Marvin Qi for stimulating discussions.
NH is supported by the Simons Collaboration on Global Categorical Symmetries and the US Department of Energy Grant 5-29073.}

\newpage
\appendix
\section{Background on Hopf algebras}\label{app:hopf}
In this appendix we review some necessary background on finite-dimensional C$^*$-Hopf algebras.
There are many good existing references that the physically or mathematically minded reader might consult, so our presentation is terse.
A collection of useful textbooks includes \cite{kassel1995quantum,montgomery1993hopf,klimyk1997quantum}. More physically minded reviews include \cite{Inamura:2021szw,lu2026generalizedkramerswannierselfdualityhopfising, kirillov2010turaevviroinvariantsextendedtqft, Kitaev:double}.

To help guide the reader unfamiliar with Hopf algebras, we note that the structure of a finite-dimensional C$^*$-Hopf algebra can be motivated entirely in terms of representation theory. A representation of an algebra $H$ on a vector space $V$ is a linear homomorphism
\begin{equation}
    \rho: H \to \End(V).
\end{equation}
A finite-dimensional C$^*$-Hopf algebra is a finite-dimensional algebra whose representations have the following properties:
\begin{enumerate}
    \item representations admit tensor products,
    \item there is a trivial representation of dimension one,
    \item representations have duals, and
    \item there is a notion of what it means for a representation to be compatible with the inner product on a Hilbert space. 
\end{enumerate}
The standard example of an algebraic object whose representations have these properties is a finite group (or equivalently, a group algebra). Thus, finite-dimensional C$^*$-Hopf algebras generalize group algebras.\footnote{This is what might be called the ``Tannakian perspective'' on Hopf algebras. Another perspective on Hopf algebras is that they generalize the structure naturally possessed by the algebra of functions on a finite group. This might be called the ``quantum group perspective'' \cite{Drinfeld:1986in}.}

Finite-dimensional C$^*$-Hopf algebras appear concretely in quantum field theories having non-invertible symmetries, where they furnish operator algebras that implement symmetry transformations \cite{cordova2024representationtheorysolitons}. The Hopf algebras acting at vertices and plaquettes in a double model can be interpreted in this way, as explained in section \ref{subsec:hopf_double}. Thus, Hopf algebras can be thought of either operationally in terms of their representations, or in terms of symmetry.

\subsection{Definitions}\label{app:hopf_def}

Each of the above ``representational'' properties is related to an additional structure with which a finite-dimensional algebra can be equipped.
A finite-dimensional C$^*$-Hopf algebra is a complex algebra $H$ having a unit together with four additional maps:
\begin{enumerate}
    \item a \textit{coproduct} \begin{equation}
        \Delta: H \to H \otimes H
    \end{equation}
    \item a \textit{counit}
    \begin{equation}
        \epsilon: H \to \C,
    \end{equation}
    \item an \textit{antipode}
    \begin{equation}
        S: H \to H,
    \end{equation}
    \item a $*$\textit{-map},
    \begin{equation}
        *: H \to H.
    \end{equation}
\end{enumerate}
The coproduct map can be conveniently expressed using ``Sweedler notation,'' an analogue of the summation convention familiar in physics. As this will be useful below, we will take a moment now to explain the notation carefully.\footnote{For a helpful discussion of Sweedler notation, see for example \cite[Chapter 3.1]{kassel1995quantum}.}

By definition, the coproduct of an element in $H$ is an element of $H \otimes H$. As such, it can be expressed as a sum of tensor products of elements in $H$. Choosing any such sum, Sweedler notation simply suppresses the summation index, writing the coproduct as
\begin{equation}
    \Delta(\alpha) = \sum_i \alpha_i^{(1)} \otimes \alpha_i^{(2)} =:  \alpha^{(1)} \otimes \alpha^{(2)}.
\end{equation}
Note that this need \textit{not} be a sum over a choice of basis elements. The notation $\alpha^{(1)} \otimes \alpha^{(2)}$ simply denotes \textit{any} presentation of the coproduct of $\alpha$ in $H \otimes H$. All axioms of the Hopf algebra are insensitive to the choice of presentation for the coproduct. In the following, any time superscripts appear on a Hopf algebra, an implicit sum is present. 

Equipped with this notation, we can now explain the compatibility properties satisfied by $(\Delta, \epsilon, S, *)$.\footnote{Since in the text we mostly leverage these identities when using Sweedler notation, we will express all conditions first in Sweedler notation and include an equivalent definition in an associated footnote.}
\begin{enumerate}
    \item The coproduct is both an algebra homomorphism and is ``coassociative,'' meaning it satisfies\footnote{This condition could also have been written
    \begin{equation*}
        (\Delta \otimes \id) \circ \Delta = (\id \otimes \Delta) \circ \Delta.
    \end{equation*}}
    \begin{equation}
        \Delta(\alpha^{(1)}) \otimes \alpha^{(2)} = \alpha^{(1)} \otimes \Delta(\alpha^{(2)}).
    \end{equation}
    This is then typically written
    \begin{equation}
        =: \alpha^{(1)} \otimes \alpha^{(2)} \otimes \alpha^{(3)}.
    \end{equation}
    Similar expressions are defined unambiguously for iterated coproducts. The coproduct can be used to define tensor product representations as\footnote{For two representations $\rho_V$ and $\rho_W$, the tensor product is equivalently the homomorphism
    \begin{equation*}
        H \xrightarrow{\Delta} H \otimes H \xrightarrow{\rho_V \otimes \rho_W} \End(V \otimes W).
    \end{equation*}}
    \begin{equation}
        \alpha \cdot (v \otimes w) = \alpha^{(1)}v \otimes \alpha^{(2)} w, \quad v \in V, w \in W,
    \end{equation}
    where $V$ and $W$ are two representations.

    \item The counit is an algebra homomorphism that defines a trivial representation as one satisfying
    \begin{equation}
        \alpha \cdot v = \epsilon(\alpha) v, \quad v \in \C.
    \end{equation}
    The counit is compatible with the coproduct as\footnote{Denoting the multiplication by $m$, this can also be expressed as
    \begin{equation*}
        m \circ (\epsilon \otimes \id)\circ \Delta = \id, \quad m \circ (\id \otimes \epsilon)\circ \Delta = \id.
    \end{equation*} }
    \begin{equation}
        \epsilon(\alpha^{(1)}) \alpha^{(2)} = \alpha^{(1)} \epsilon(\alpha^{(2)}) = \alpha.
    \end{equation}
    We often use this identity applied to iterated coproducts in the main text.

    \item The antipode is a (linear) algebra anti-homomorphism and coalgebra anti-homomorphism, meaning
    \begin{equation}
        S(\alpha\beta) = S(\beta)S(\alpha), \quad S(\alpha)^{(1)} \otimes S(\alpha)^{(2)} = S(\alpha^{(2)})\otimes S(\alpha^{(1)}),
    \end{equation}
    that satisfies the compatibility condition\footnote{This is equivalently the condition,
    \begin{equation*}
        m \circ (S \otimes \id) \circ \Delta = \epsilon, \quad m \circ (\id \otimes S) \circ \Delta = \epsilon.
    \end{equation*}
    }
    \begin{equation}
        S(\alpha^{(1)})\alpha^{(2)} = \alpha^{(1)}S(\alpha^{(2)}) = \epsilon(\alpha).
    \end{equation}
    This defines the notion of dual representation as
    \begin{equation}
        \alpha \cdot \lambda = \lambda \circ S(\alpha), \quad \lambda \in V^\vee.
    \end{equation}
    Like the inverse in a group, if an antipode exists, it is unique.

    \item Finally, the $*$-map is an (antilinear) algebra anti-homomorphism and coalgebra homomorphism, meaning
    \begin{equation}
        (\alpha\beta)^* = \beta^* \alpha^*, \quad (\alpha^*)^{(1)} \otimes (\alpha^*)^{(2)} = (\alpha^{(1)})^*\otimes (\alpha^{(2)})^*.
    \end{equation}
    The map is also involutive:
    \begin{equation}
        (\alpha^*)^* = \alpha.
    \end{equation}
    The $*$-map defines the notion of a unitary representation (or $*$-representation) as a representation satisfying\footnote{Note that it is \textit{not} required that the elements of the Hopf algebra be represented by unitary matrices -- in general elements don't even have inverses. Unitarity is instead the statement that the $*$-map be realized as the adjoint.}
    \begin{equation}
        (\alpha^*) \cdot v = \alpha^\dagger v.
    \end{equation}
\end{enumerate}
This exhausts all compatibility conditions that the maps must satisfy. Along with these conditions, the algebra is also required to be C$^*$ with respect to the $*$-map, meaning that it admits a faithful unitary representation.

It will be useful going forward to know three properties of the antipode of a finite-dimensional C$^*$-Hopf algebra. First, it is involutive \cite{larson}:
\begin{equation}
    S(S(\alpha)) = \alpha.
\end{equation}
Second, the counit and antipode satisfy
\begin{equation}
    \epsilon(S(\alpha)) = \epsilon(\alpha).
\end{equation}
Finally, the composite map with the $*$-structure is involutive \cite[Section 1.2.7]{klimyk1997quantum}:
\begin{equation}\label{eq:CPT_map}
    \gamma = S \circ *, \quad \gamma(\gamma(\alpha)) = \alpha.
\end{equation}

In the bulk text, we at times make use of the property of cocommutativity as well. An element $\alpha$ of a Hopf algebra is called \textit{cocommutative} if it satisfies,
\begin{equation}
    \alpha^{(1)} \otimes \alpha^{(2)} = \alpha^{(2)} \otimes \alpha^{(1)}.
\end{equation}
A cocommutative element has the special property that its iterated coproducts are left unchanged under \textit{cyclic} permutations. If we denote the $n$-th iterated coproduct as $\Delta^{n}$, then cocommutativity implies
\begin{multline}
    \alpha^{(1)} \otimes \alpha^{(2)} \otimes \cdots \otimes \alpha^{(n+1)} = \alpha^{(1)} \otimes \Delta^{n-1}(\alpha^{(2)}) = \alpha^{(2)} \otimes \Delta^{n-1}(\alpha^{(1)}) \\ = \alpha^{(n+1)} \otimes \alpha^{(1)} \otimes \cdots \otimes \alpha^{(n)}.
\end{multline}

\subsection{Dual Hopf algebras}
To every finite-dimensional C$^*$-Hopf algebra there is an associated \textit{dual} Hopf algebra. Given a Hopf algebra $H$, the dual Hopf algebra is defined on the dual vector space and is denoted $H^\vee$. Its product, coproduct, unit, counit, and antipode are defined using the dual maps of $H$: for $f$ and $g$ elements of $H^{\vee}$, one defines
\begin{equation}
    (f \cdot g)(\alpha) = f(\alpha^{(1)}) g(\alpha^{(2)}), \quad \Delta_\vee(f)(\alpha \otimes \beta) = f(\alpha\beta),
\end{equation}
and
\begin{equation}\label{eq:dual_unit_counit}
    1_\vee(\alpha) = \epsilon(\alpha), \quad \epsilon_\vee(f) = f(1), \quad S_\vee(f) = f \circ S.
\end{equation}
The $*$-map, by contrast, is \textit{not} defined by the pullback, but rather as
\begin{equation}\label{eq:dual_star_map}
    f^*(\alpha) = \overline{f(S(\alpha)^*)},
\end{equation}
where the overline denotes complex conjugation in $\C$. This is necessary to produce a $*$-map with the correct compatibility properties with the other maps \cite[Section 1.2.7]{klimyk1997quantum}.

\begin{example}
    The primary concrete example studied in this text is the dual group algebra of a finite group $G$. This is the dual Hopf algebra of the group algebra $\C[G]$.
    The Hopf algebra $\mathbb{C}[G]$ is the algebra of formal linear combination of group elements, equipped with the maps
    \begin{equation}\label{eq:grp_hopf_data}
        \Delta(g) = g \otimes g, \quad \epsilon(g) = 1, \quad S(g) = g^{-1}, \quad g^* = g^{-1},
    \end{equation}
    extended linearly or antilinearly as necessary. Taking the duals of these maps then provides the maps used in the main text:
    \begin{equation}
        p_g \cdot p_h = \delta_{g,h} p_g, \quad \Delta(p_g) = \sum_{h,k} \delta_{g,hk} p_h \otimes p_k, \quad \epsilon(p_g) = \delta_{1,g}, \quad S(p_g) = p_{g^{-1}}, \quad p_{g}^* = p_g.
    \end{equation}
    Note that the $*$-structure acts on general functions as
    \begin{equation}
        f^* = \sum_{g \in G}\overline{f(g)} (p_g)^* = \sum_{g \in G}\overline{f(g)} p_g,
    \end{equation}
    meaning that a function is sent to its complex conjugate.
\end{example}

\subsection{Representation by duals}
A Hopf algebra always carries a left and right representation on itself by multiplication
\begin{equation}\label{eq:rep_self_on_self}
    L_\beta (\alpha) = \beta \alpha, \quad R_\beta (\alpha) = \alpha S(\beta), \quad \alpha, \beta \in H.
\end{equation}
As in the main text, we introduce an antipode factor for the right action so to obtain a representation (i.e.\ a left module) in the variable $\beta$. The dual space also carries a left and right $H$-representation by pre-composition with the left and right actions\footnote{It is helpful to think of $H^\vee$ as the analogue of $L^2(G)$, for $G$ a compact Lie group. The actions above are then analogous to the natural left and right $G$-representations on $L^2(G)$.}
\begin{equation}
    (L_\beta f)(\alpha) = f(\alpha\beta), \quad (R_\beta f)(\alpha) = f(S(\beta) \alpha), \quad f \in H^\vee.
\end{equation}
This representation can also be thought of as being induced by the coproduct of $H^\vee$.
Concretely, the coproduct of $H^{\vee}$ is defined by,
\begin{equation}
    f^{(1)}(\alpha)f^{(2)}(\beta) = f(\alpha\beta),
\end{equation}
and therefore one may rewrite the left- and right- actions of $H$ on $H^{\vee}$ as
\begin{equation}
    L_\beta f = f^{(2)}(\beta)f^{(1)}, \quad R_\beta f = f^{(1)}(S(\beta))f^{(2)}, \quad f \in H^\vee.
\end{equation}
Treating the initial Hopf algebra $H$ as the dual of $H^{\vee},$ one obtains a natural left and right representation of $H^{\vee}$ on $H$ as
\begin{equation}\label{eq:rep_dual_on_self}
    L_f (\alpha) = \alpha^{(1)}f(\alpha^{(2)}),\quad R_f (\alpha) = S(f)(\alpha^{(1)})\alpha^{(2)}, \quad f \in H^\vee.
\end{equation}
For the Hopf-algebraic double models studied in the main text, equations \eqref{eq:rep_self_on_self} and \eqref{eq:rep_dual_on_self} are the four natural representations of $H$ and $H^{\vee}$ acting on the edge Hilbert spaces of the model.

\subsection{Peter-Weyl factorization}\label{app:peter-weyl}
A finite-dimensional C$^*$-Hopf algebra lacks a canonical choice of basis. Instead, it admits two natural decompositions.

The first decomposition, called the Wedderburn-Artin decomposition, follows from the C$^*$-algebra structure of the Hopf algebra. The simplest kind of C$^*$-algebra is a matrix algebra $\mathrm{Mat}(d,\mathbb{C}),$ which has a unique irreducible representation $\mathbb{C}^d$. More generally, any finite-dimensional C$^*$-algebra can be written as a direct sum of matrix algebras, and consequently, can be expressed as a direct sum over the endomorphism algebras of its irreducible representations:
\begin{equation}
    \begin{aligned}
        H &\cong \bigoplus_x \End(W_x), \quad H^{\vee} \cong \bigoplus_a \End(V_a),
    \end{aligned}
\end{equation}
where each $W_x$ ($V_a$) is an irreducible representation of $H$ ($H^{\vee}$). This decomposition of algebras is the Wedderburn-Artin decomposition. Representationally, this can be viewed as a decomposition of a C$^*$ algebra into irreps of its natural left and right actions on itself.

The second natural decomposition is called the Peter-Weyl decomposition. This decomposition follows from the so called C$^*$\textit{-coalgebra} structure of the Hopf algebras; that is, the data of the coproduct and counit.
Alternatively, the Peter-Weyl decomposition can be viewed as arising from the left and right actions of $H$ and $H^\vee$ on \textit{one another}. Here we will review this latter perspective.

The Peter-Weyl decomposition takes takes the form,
\begin{equation}
    H \cong \bigoplus_a \End(V_a), \quad  H^{\vee} \cong \bigoplus_x \End(W_x),
\end{equation}
with these symbols ``$\cong$'' denoting isomorphisms of $H^\vee$-representations and $H$-representations, respectively.\footnote{Use of the term ``Peter-Weyl'' for these decompositions is motivated by analogy to the case of a compact Lie group, where the Peter-Weyl theorem decomposes the Hilbert space $L^2(G)$ as 
\begin{equation*}
    L^2(G) \cong \bigoplus_a \End(V_a) \cong \bigoplus_a V_a^\vee \otimes V_a.
\end{equation*}
} In particular, unlike the Wedderburn-Artin decomposition, the congruency symbols above do \textit{not} denote algebra isomorphisms.\footnote{Instead, they are isomorphisms of \textit{coalgebras}. For further discussion see for example \cite{montgomery1993hopf}.}

The Peter-Weyl decomposition is constructed for the dual Hopf algebra as follows.\footnote{The decomposition for the Hopf algebra itself follows from this analysis using the identity $H = (H^\vee)^\vee$.} For an irreducible $H$-representation $W_x$, define the functionals
\begin{equation}
    M_x(\alpha) = \tr_x(M_x \rho_x(\alpha)), \quad M_x \in \End(W_x).
\end{equation}
This produces a map
\begin{equation}
    \End(W_x) \to H^\vee,
\end{equation}
which one can verify is an inclusion. The spaces $\End(W_x)$ naturally possess a left and right $H$-action induced by the representation:
\begin{equation}
    L_\beta M_x = \rho_x(\beta)M_x, \quad R_\beta M_x = M_x\rho_x(S(\beta)).
\end{equation}
If follows that the inclusion of $\End(W_x)$ into $H^{\vee}$ is equivariant with respect to both actions, since $H$ acts on functionals as
\begin{equation}
    (L_\beta M_x)(\alpha) = M_x(\alpha \beta) = \tr_x(M_x \rho_x(\alpha\beta)) = \tr_x(\rho_x(\beta)M_x\rho_x(\alpha)) = (\rho_x(\beta)M_a)(\alpha),
\end{equation}
and similarly for the right action. Taken together, we therefore have an inclusion of representations,
\begin{equation} \label{eq:end-Hvee-map}
    \bigoplus_x \End(W_x) \to H^\vee.
\end{equation}
Finally, because $H$ and $H^{\vee}$ have the same dimension, and because the Wedderburn-Artin decomposition implies the identity
\begin{equation}
    \dim(H) = \sum_x \dim(\End(V_x)),
\end{equation}
the map \eqref{eq:end-Hvee-map} is surjective and hence an isomorphism.

Applying the Peter-Weyl decomposition to a Hopf algebra $H$, one obtains
\begin{equation}
    H \cong \bigoplus_a \End(V_a),
\end{equation}
where each $V_a$ is an irreducible representation of $H^{\vee}$. A choice of orthonormal basis $|j\rangle_a$ on each irreducible representation $V_a$ produces a natural basis for $\End(V_a),$ 
\begin{equation}
    Z^a_{ij} \in \End(V_a), \quad Z_{ij}^a |k\rangle_a = \delta_{i,k} |j\rangle_a,
\end{equation} 
or simply
\begin{equation}
    Z^a_{ij} = |j\rangle\langle i|_a.
\end{equation}
By embedding $Z^a_{ij}$ into $H$ via the Peter-Weyl decomposition, we may act with $Z^{a}_{ij}$ on any element of $\lambda \in H^{\vee}$ to obtain 
\begin{equation}
    Z_{ij}^a(\lambda) = \tr_a(Z_{ij}^a \rho_a(\lambda)) = \langle i | \rho_a(\lambda) |j\rangle_a = \rho(\lambda)_{ij}^a.
\end{equation}
These basis elements have an especially simple coproduct.
\begin{proposition}
    The coproduct of the $Z^a_{ij}$ basis elements takes the form
    \begin{equation}
        \Delta(Z_{ij}^a) = \sum_{k} Z^a_{ik} \otimes Z_{k j}^a.
\end{equation}
\end{proposition}
\begin{proof}
    We will show this by direct computation. The important property that the operators $Z_{ij}^a$ satisfy is that they are dual to the basis of \textit{idempotes} in the dual Hopf algebra. This basis is the following. 
    
    From the Wedderburn-Artin decomposition, the dual Hopf algebra can be decomposed as
    \begin{equation}
        H^\vee \cong \bigoplus_a \End(V_a).
    \end{equation}
    Having fixed the choice of orthonormal basis on each irreducible representation in defining $Z_{ij}^a$, using this same basis there is a distinguished set of operators in the dual Hopf algebra
    \begin{equation}
        \Theta^a_{ij} \in H^\vee,
    \end{equation}
    defined to be represented in an irreducible representation $V_a$ as 
    \begin{equation}
        \rho_a(\Theta^a_{ij}) = \ket{i}\bra{j}_a \in \End(V_a).
    \end{equation}
    This is a basis of \textit{orthogonal idempotes} in the sense that they satisfy,
    \begin{equation}
        \Theta^a_{ij}\Theta^a_{k\ell} = \delta^{ab}\delta_{jk}\Theta^a_{i\ell}.
    \end{equation}

    Unpacking definitions, we have
    \begin{equation}
        Z_{ij}^a(\Theta_{k\ell}^b) = \tr_a(\ket{j}\bra{i}_a \rho_a(\Phi^b_{k\ell})) = \delta^{ab}\tr_a(\ket{j}\bra{i}_a \ket{k}\bra{\ell}_a) = \delta^{ab}\delta_{ik}\delta_{j\ell}.
    \end{equation}
    The two bases are therefore dual. Using this fact, the coproduct can then be computed directly:
    \begin{equation}
        \Delta(Z^a_{ij}) (\Theta^b_{k \ell} \otimes \Theta^c_{m n})
            = Z^a_{ij} (\Theta^b_{k \ell} \Theta^c_{m n})
            = \delta^{b c} \delta_{\ell m} Z^a_{ij}(\Theta^b_{k n})
            = \delta^{bc} \delta^{ab} \delta_{\ell m} \delta_{i k} \delta_{j n}.
    \end{equation}
    This implies the proposition, as one can separately compute
    \begin{equation}
        (\sum_{p} Z^a_{ip} \otimes Z^a_{pj}) (\Theta^b_{k \ell} \otimes \Theta^c_{m n})
        = \sum_p \delta^{a b} \delta^{a c} \delta_{i k} \delta_{p \ell} \delta_{p m} \delta_{j n}
        = \delta^{ab} \delta^{bc} \delta_{i k} \delta_{\ell m} \delta_{j n},
    \end{equation}
    which matches.
\end{proof}

This proposition means that the basis elements $Z_{ij}^a$ form a so-called \textit{matrix comatrix algebra}. It follows from this observation that the dual Hopf algebra acts on these elements, using the definition \eqref{eq:rep_dual_on_self}, as
\begin{equation}
    \begin{aligned}
        L_\lambda \ket{Z_{ij}^a} &= \sum_k \lambda(Z_{k j}^a)\ket{Z_{ik}^a} = \sum_k \rho(\lambda)_{kj}^a \ket{Z_{ik}^a} \\
        R_\lambda \ket{Z_{ij}^a} &= \sum_k \lambda(Z_{ik}^a)\ket{Z_{k j}^a} = \sum_k \rho(\lambda)_{ik}^a \ket{Z_{kj}^a}.
    \end{aligned}
\end{equation}
Another useful property of the basis elements $Z_{ij}^a$ is that they satisfy the identity:
\begin{equation}\label{eq:unitary_coalg}
    S(Z_{ij}^a)^* = Z_{ji}^a,
\end{equation}
which follows from definition \eqref{eq:dual_star_map} and property \eqref{eq:CPT_map} and the manipulation 
\begin{equation}
    S(Z_{ij}^a)^*(\lambda) = \overline{S(Z_{ij}^a)(S(\lambda)^*)} = \overline{\rho(\lambda^*)_{ij}^a} = \rho(\lambda)_{ji}^a = Z_{ji}^a(\lambda), \quad \lambda \in H^\vee.
\end{equation}
This property means that the basis elements $Z_{ij}^a$ form a \textit{unitary comatrix algebra}. 

\subsection{Integrals and Hilbert space structure} \label{app:haar}
A finite-dimensional C$^*$-Hopf algebra has a distinguished element called its \textit{Haar integral}. This is the unique element $h \in H$ satisfying
\begin{equation}
    \alpha h = h \alpha = \epsilon(\alpha) h, \quad \epsilon(h) = 1, \quad \alpha \in H.
\end{equation}
The Haar integral on the dual Hopf algebra is also called the \textit{Haar measure} on $H$, and we denote it in this text by $\mu \in H^\vee.$

The Haar integral on a general Hopf algebra $H$ plays the same role as integration over the compact-$G$-Haar measure for a function algebra $L^2(G)$.
Namely, $h$ defines an invariant integration of functions in $H^\vee$ via evaluation:
\begin{equation}
    \int : H^\vee \to \C, \quad \int f = f(h).
\end{equation}
This integration is invariant in the sense that one has
\begin{equation}
    h(L_\alpha f) = h(R_\alpha f) = \epsilon(\alpha)h(f).
\end{equation}
For the $L_{\alpha}$ action, this follows from the chain of identities
\begin{equation}
    h(L_\alpha f) = f^{(2)}(\alpha)h(f^{(1)}) = f^{(1)}(h)f^{(2)}(\alpha) = f(h\alpha) = \epsilon(\alpha) f(h) = \epsilon(\alpha)h(f),
\end{equation}
and a similar chain holds for $R_{\alpha}$.
The term ``invariant'' is used in the sense that left and right multiplication produce only the singlet action of $H$. 

Thought of as an invariant integration on $H^\vee$, the Haar integral is used to define an inner product on $H^\vee$ as
\begin{equation}
    \langle f, g \rangle := h(f^* g), \quad f,g \in H^\vee.
\end{equation}
For the special case of a finite group, the Haar integral takes the form
\begin{equation}
    h = \frac{1}{|G|} \sum_g g,
\end{equation}
which is simply the invariant normalized integral on $G$. More generally, we have the following proposition.
\begin{proposition}
The Haar integral can be written as a weighted sum over characters in a finite-dimensional C$^*$-Hopf algebra,
    \begin{equation}\label{eq:haar_int_decomp}
        h = \frac{1}{|H|} \sum_a \dim(V_a)\chi_a,
    \end{equation}
    where $V_a$ are irreducible representations of $H^\vee$.
\end{proposition}
\noindent Here we give a representation-theoretic proof of this result.  Conceptually, it will be slightly easier to prove this proposition for the Haar measure, so we will prove the dual statement
\begin{equation}\label{eq:Haar_measure_decomp}
    \mu = \frac{1}{|H|} \sum_x \dim(W_x) \chi_x,
\end{equation}
with $W_x$ an irreducible representation of $H$. 

\begin{proof}
    To prove that $\mu$ as defined in \eqref{eq:Haar_measure_decomp} is the Haar measure, we must show
    \begin{equation}\label{eq:pf_haar_measure}
	   \mu f = f \mu = \epsilon_{\vee}(f) \mu, \quad f \in H^\vee,
    \end{equation}
    and
    \begin{equation}
	   \epsilon_{\vee}(\mu) = 1.
    \end{equation}
    The second condition is satisfied by the choice of overall normalization factor,
    \begin{equation}
        \epsilon_{\vee}(\mu) = \frac{1}{|H|}\sum_x \dim(W_x) \epsilon_{\vee}(\chi_x) = \frac{1}{|H|}\sum_x \dim(W_x)^2 = 1,
    \end{equation}
    where we've used that by definition \eqref{eq:dual_unit_counit}, the counit of a character is 
    \begin{equation}
        \epsilon_{\vee}(\chi_x) = \chi_x(1) = \dim(W_x).
    \end{equation}
    It remains therefore to prove the conditions in equation \eqref{eq:pf_haar_measure}.
    To do so, we will consider the Peter-Weyl decomposition of $H^\vee$ and show that \eqref{eq:pf_haar_measure} holds on each summand individually. 
    
    Let $W_x$ be an irreducible representation of $H$. Our goal is to show
    \begin{equation}
        \mu M_x = M_x \mu = \mu \epsilon_\vee(M_x), \quad M_x \in \End(W_x).
    \end{equation}
    Note first that, by the definition of the Peter-Weyl decomposition, the counit evaluated on an endomorphism is its trace:
    \begin{equation}
        \epsilon_\vee(M_x) = M_x(1) =\tr_x(M_x).
    \end{equation}
    We are therefore aiming to show
    \begin{equation}\label{eq:exp_to_prove}
        (\mu M_x)(\alpha) = (M_x \mu)(\alpha) = \mu(\alpha) \tr_x(M_x), \quad \alpha \in H,
    \end{equation}
    which we will accomplish by direct computation.    
    
    Our primary tools will be some simple aspects of representation theory. The key point is that our candidate expression for $\mu$ is the (normalized) character of the regular representation; that is, the $H$ representation induced by the left action of $H$ on itself:\footnote{When emphasizing its representational structure, we will denote the regular representation by $W_{\reg}$.}
    \begin{equation}
        \mu = \frac{1}{|H|}\chi_{\reg}.
    \end{equation}
    In \eqref{eq:exp_to_prove} we can therefore write the final term as\footnote{To ease notation, we will not distinguish $\alpha$ from its representing matrix $\rho(\alpha)$ when the representation is clear from context.}
    \begin{equation}
        \mu(\alpha) \tr_x(M_x) = \frac{1}{|H|}\chi_{\reg}(\alpha)\tr_x(M_x) = \frac{1}{|H|}\tr_{\reg}(\alpha)\tr_{x}(M_x) = \frac{1}{|H|}\tr_{\reg \otimes x}(\alpha \otimes M_x).
    \end{equation}
    Similarly, unpacking definitions, the first two expressions in \eqref{eq:exp_to_prove} can also be written as traces:
    \begin{align}\label{eq:trace_id_one}
    \begin{split}
        (\mu M_x)(\alpha) = \mu(\alpha^{(1)})M_x(\alpha^{(2)}) & = \frac{1}{|H|}\tr_{\reg}(\alpha^{(1)})\tr_{x}(M_x \alpha^{(2)}) \\ & = \frac{1}{|H|}\tr_{\reg \otimes x}((1 \otimes M_x)(\alpha^{(1)} \otimes \alpha^{(2)}))
    \end{split}
    \end{align}
    and
    \begin{align}\label{eq:trace_id_two}
        \begin{split}
        (M_x\mu)(\alpha) & = M_x(\alpha^{(1)})\mu(\alpha^{(2)}) = \frac{1}{|H|}\tr_{\reg}(\alpha^{(2)})\tr_{x}(M_x \alpha^{(1)}) \\ & = \frac{1}{|H|}\tr_{\reg \otimes x}((1 \otimes M_x)(\alpha^{(2)} \otimes \alpha^{(1)})).
        \end{split} 
    \end{align}
    We therefore have two trace identities to prove in the representation $W_{\reg} \otimes W_x$.
    
    To show the equality of the three above expressions, we will make use of the following special property of the regular representation: for any choice of irreducible representation $W_x$, the regular representation satisfies
    \begin{equation}
        W_{\reg} \otimes W_x \cong \bigoplus_{i = 1}^{\dim(W_x)} W_{\reg}.
    \end{equation}
    In fact, there is a natural choice of such an isomorphism. This is defined as a map
    \begin{equation}
        \Phi: W_{\reg} \otimes (W_x)_{\triv} \to W_{\reg} \otimes W_x.
    \end{equation}
    Here $(W_x)_{\triv}$ denotes the vector space $W_x$ equipped with the trivial $H$-representation, meaning that a Hopf algebra element acts on an element of $W_{\reg}\otimes (W_{x})_{\triv}$ as\footnote{We will use $\rho_{x'}(\alpha)$ to denote the element $\alpha$ acting on the vector space $W_x$ equipped with the trivial representation.}
    \begin{equation}
        \rho_{\reg \otimes x'}(\alpha) (\beta \otimes v) = \rho_{\reg}(\alpha^{(1)})\beta \otimes \rho_{x'}(\alpha^{(2)})v = \alpha^{(1)}\beta \otimes \epsilon(\alpha^{(2)})v = \alpha\beta \otimes v,
    \end{equation}
    where for the moment we will restore explicit representation labels. The isomorphism and its inverse are defined to act on elements as
    \begin{equation}
        \Phi(\beta \otimes v) = \beta^{(1)} \otimes \rho_x(\beta^{(2)})v, \quad \Phi^{-1}(\beta \otimes v) = \beta^{(1)} \otimes \rho_x(S(\beta^{(2)}))v.
    \end{equation}
    Importantly, this map is equivariant: in one direction,
    \begin{equation}
        \Phi(\rho_{\reg \otimes x'}(\alpha) (\beta \otimes v)) = \Phi(\alpha\beta \otimes v) = (\alpha\beta)^{(1)} \otimes \rho_{x}((\alpha\beta)^{(2)})v,
    \end{equation}
    and in the other,
    \begin{multline}
        \rho_{\reg \otimes x}(\alpha) \Phi(\beta \otimes v) = \rho_{\reg \otimes x}(\alpha) (\beta^{(1)} \otimes \rho_x(\beta^{(2)})v) \\ = \rho_{\reg}(\alpha^{(1)})\beta^{(1)} \otimes \rho_x(\alpha^{(2)})\rho_x(\beta^{(2)})v = (\alpha\beta)^{(1)} \otimes \rho_x((\alpha\beta)^{(2)})v,
    \end{multline}
    which are equal. Therefore,
    \begin{equation}
        \rho_{\reg \otimes x}(\alpha) \Phi = \Phi \rho_{\reg \otimes x'}(\alpha),
    \end{equation}
    and so this map is indeed an isomorphism of representations. 

    We can use the isomorphism $\Phi$ to compute the trace in \eqref{eq:trace_id_one} and \eqref{eq:trace_id_two}. Equivariance of $\Phi$ provides an equality of operators,
    \begin{equation}
        \rho_{\reg}(\alpha^{(1)}) \otimes \rho_x(\alpha^{(2)}) = \Phi(\rho_{\reg}(\alpha) \otimes 1)\Phi^{-1}.
    \end{equation}
    which allows for the first trace to be rewritten as\footnote{Again now suppressing explicit representation labels for brevity.}
    \begin{equation}
        \begin{aligned}
            \frac{1}{|H|}\tr_{\reg \otimes x}((1 \otimes M_x)(\alpha^{(1)} \otimes \alpha^{(2)})) &= \frac{1}{|H|}\tr_{\reg \otimes x}((1 \otimes M_x) \Phi (\alpha \otimes 1) \Phi^{-1}) \\
            &= \frac{1}{|H|}\tr_{\reg \otimes x'}(\Phi^{-1} (1 \otimes M_x) \Phi (\alpha \otimes 1)),
        \end{aligned}
    \end{equation}
    where the latter trace is now taken in the representation $W_{\reg} \otimes (W_{x})_{\triv}$. Unpacking definitions, the conjugated $1 \otimes M_x$ acts on an element as,
    \begin{equation}
        \Phi^{-1}(1 \otimes M_x)\Phi(\beta \otimes v) = \beta^{(1)} \otimes S(\beta^{(2)})M_x \beta^{(3)}v.
    \end{equation}
    When taking the trace over $(W_x)_{\triv}$ first, we will therefore get the total operator on $W_{\reg}$:
    \begin{equation}
        \beta \mapsto \alpha \beta^{(1)} \tr_x(S(\beta^{(2)})M_x \beta^{(3)}), \quad \beta \in W_{\reg}.
    \end{equation}
    By cyclicity of the trace, this is the element
    \begin{equation}
        = \alpha \beta^{(1)} \tr_x(\beta^{(3)}S(\beta^{(2)})M_x).
    \end{equation}
    Moreover, since in a finite dimensional C$^*$-Hopf algebra, the antipode is involutive, we have
    \begin{equation}
        \beta^{(3)}S(\beta^{(2)}) = S(\beta^{(2)}S(\beta^{(3)})) = S(\epsilon(\beta^{(2)})) = \epsilon(\beta^{(2)})S(1) = \epsilon(\beta^{(2)}).
    \end{equation}
    Therefore, this is the element,
    \begin{equation}
        = \alpha \beta^{(1)}\epsilon(\beta^{(2)}) \tr_x(M_x) = \alpha\beta \tr_x(M_x).
    \end{equation}
    Taking the remaining trace over the regular representation, we have therefore shown,
    \begin{equation}
        \frac{1}{|H|}\tr_{\reg \otimes x}((1 \otimes M_x)(\alpha^{(1)} \otimes \alpha^{(2)})) = \frac{1}{|H|}\tr_{\reg}(\alpha)\tr_{x}(M_x),
    \end{equation}
    which is precisely the equality that we were aiming to demonstrate. An identical argument holds for computing the second trace identity \eqref{eq:trace_id_two}. We have therefore demonstrated equation \eqref{eq:pf_haar_measure}, which proves that $\mu$ is indeed the Haar measure.
\end{proof}

Using equation \eqref{eq:haar_int_decomp} as an explicit expression for the Haar integral, the inner product on $H^\vee$ may be written as
\begin{equation}
    \langle f, g \rangle = \frac{1}{|H|}\sum_{x} d_x \chi_x(f^* g) = \frac{1}{|H|}\sum_{x,i,j} d_x f^*(Z_{ij}^x) g(Z_{ji}^x),
\end{equation}
which further simplifies to
\begin{equation}
\begin{aligned}
    & = \frac{1}{|H|}\sum_{x,i,j} d_x \overline{f(S(Z_{ij}^x)^*)} g(Z_{ji}^x) \\
    & = \frac{1}{|H|}\sum_{x,i,j} d_x \overline{f(Z_{ji}^x)} g(Z_{ji}^x).
\end{aligned}
\end{equation}
This final expression shows that this inner product is in fact the immediate generalization of the natural inner product on functions on a finite group.
\begin{corollary}
    The Haar inner product on a dual algebra $H^\vee$ can be concretely expressed as
    \begin{equation}
        \langle f, g \rangle = \frac{1}{|H|}\sum_{x,i,j} d_x \overline{f(Z_{ji}^x)} g(Z_{ji}^x),
    \end{equation}
    with $x$ ranging over irreducible representations of $H$ and $i,j$ being indices in an orthonormal basis of the representation.
\end{corollary}
Another useful expression for the inner product of the basis elements $Z_{ij}^a$ is given in the following proposition.
\begin{proposition}
    The inner product of two basis elements $Z_{ij}^a, Z_{k\ell}^b$ takes the form
    \begin{equation}
        \langle Z_{ij}^a, Z_{k\ell}^b \rangle = \frac{\delta^{ab}}{d_a}\delta_{ik}\delta_{j\ell},
    \end{equation}
    where $d_a$ is the dimension of the irreducible $H^\vee$-representation $V_a$.
\end{proposition}
\begin{proof}
    By definition, one has
    \begin{equation}
        \begin{aligned}
        \langle Z_{ij}^a, Z_{k\ell}^b \rangle &= \mu( (Z_{ij}^a)^* Z_{k\ell}^b) = (Z_{ij}^a)^*(\mu^{(1)})Z_{k\ell}^b(\mu^{(2)}) = \rho(S(\mu^{(1)}))_{ji}^a \rho(\mu^{(2)})^b_{k\ell}.
        \end{aligned}
    \end{equation}
    In an orthonormal basis,
    \begin{equation}
        \rho(S(\lambda))_{ji}^a = \rho(\lambda)_{ij}^{\overline{a}}, \quad \lambda \in H^\vee,
    \end{equation}
    and so the inner product simplifies to
    \begin{equation}
        \langle Z_{ij}^a, Z_{k\ell}^b \rangle= \rho(\mu)_{ik,j\ell}^{\overline{a}\otimes b}.
    \end{equation}
    Since $\mu$ is the singlet projector, this matrix element will vanish unless $V_a = V_b$. Furthermore, for this case the matrix element takes the form,
    \begin{equation}
        \rho(\mu)_{ik,j\ell}^{\overline{a}\otimes b} = \frac{\delta^{ab}}{d_a}\delta_{ik}\delta_{j\ell}.
    \end{equation}
    This gives the stated expression, which is used in the main text.
\end{proof}
\noindent Finally, in the bulk of the text the notion of the neutral component of an operator,
\begin{equation}
    \cO_0 = h^{(1)}\cO S(h^{(2)}),
\end{equation}
plays an important role. Using the explicit form of the Haar integral just derived, this takes the more concrete form,
\begin{equation}
    \cO_0 = \frac{1}{|H|}\sum_a d_a Z_{ij}^a \cO S(Z_{ji}^a).
\end{equation}
Using the identity \eqref{eq:unitary_coalg}, this can be re-expressed as,
\begin{equation}
    \cO_0 = \frac{1}{|H|}\sum_a d_a Z_{ij}^a \cO (Z_{ij}^a)^*,
\end{equation}
which appears in the main text as equation \eqref{eq:neutral_comp_explicit}.

\newpage
\bibliographystyle{JHEP}
\bibliography{biblio.bib}

\providecommand{\href}[2]{#2}\begingroup\raggedright\begin{thebibliography}{10}

\bibitem{Haag-Kastler}
R.~Haag and D.~Kastler, {\it An algebraic approach to quantum field theory},
  {\em {Journal of Mathematical Physics}} {\bf 5} (1964), no.~7 848--861.

\bibitem{DHR-1}
S.~Doplicher, R.~Haag, and J.~E. Roberts, {\it {Local observables and particle
  statistics. 1}},  {\em Commun. Math. Phys.} {\bf 23} (1971) 199--230.

\bibitem{DHR-2}
S.~Doplicher, R.~Haag, and J.~E. Roberts, {\it {Local observables and particle
  statistics. 2}},  {\em Commun. Math. Phys.} {\bf 35} (1974) 49--85.

\bibitem{Fredenhagen:braids-1}
K.~Fredenhagen, K.-H. Rehren, and B.~Schroer, {\it {Superselection Sectors with
  Braid Group Statistics and Exchange Algebras. 1. General Theory}},  {\em
  Commun. Math. Phys.} {\bf 125} (1989) 201.

\bibitem{Fredenhagen:braids-2}
K.~Fredenhagen, K.-H. Rehren, and B.~Schroer, {\it {Superselection sectors with
  braid group statistics and exchange algebras. 2. Geometric aspects and
  conformal covariance}},  {\em Rev. Math. Phys.} {\bf 4} (1992), no.~spec01
  113--157.

\bibitem{Longo:subfactors}
R.~Longo and K.-H. Rehren, {\it {Nets of subfactors}},  {\em Rev. Math. Phys.}
  {\bf 7} (1995) 567--598, [\href{http://arxiv.org/abs/hep-th/9411077}{{\tt
  hep-th/9411077}}].

\bibitem{Naaijkens:index}
P.~Naaijkens, {\it {Kosaki-Longo index and classification of charges in 2D
  quantum spin models}},  {\em J. Math. Phys.} {\bf 54} (2013) 081901,
  [\href{http://arxiv.org/abs/1303.4420}{{\tt arXiv:1303.4420}}].

\bibitem{Naaijkens:subfactors}
P.~Naaijkens, {\it {Subfactors and quantum information theory}},  {\em Contemp.
  Math.} {\bf 717} (2018) 257--279,
  [\href{http://arxiv.org/abs/1704.05562}{{\tt arXiv:1704.05562}}].

\bibitem{Casini:order}
H.~Casini, M.~Huerta, J.~M. Magan, and D.~Pontello, {\it {Entropic order
  parameters for the phases of QFT}},  {\em JHEP} {\bf 04} (2021) 277,
  [\href{http://arxiv.org/abs/2008.11748}{{\tt arXiv:2008.11748}}].

\bibitem{Casini:completeness}
H.~Casini and J.~M. Magan, {\it {On completeness and generalized symmetries in
  quantum field theory}},  {\em Mod. Phys. Lett. A} {\bf 36} (2021), no.~36
  2130025, [\href{http://arxiv.org/abs/2110.11358}{{\tt arXiv:2110.11358}}].

\bibitem{Jones:boundary}
C.~Jones, P.~Naaijkens, D.~Penneys, and D.~Wallick, {\it {Local topological
  order and boundary algebras}},  \href{http://arxiv.org/abs/2307.12552}{{\tt
  arXiv:2307.12552}}.

\bibitem{Benedetti:modular}
V.~Benedetti, H.~Casini, Y.~Kawahigashi, R.~Longo, and J.~M. Magan, {\it
  {Modular invariance as completeness}},  {\em Phys. Rev. D} {\bf 110} (2024),
  no.~12 125004, [\href{http://arxiv.org/abs/2408.04011}{{\tt
  arXiv:2408.04011}}].

\bibitem{Shao:additivity}
S.-H. Shao, J.~Sorce, and M.~Srivastava, {\it {Additivity, Haag duality, and
  non-invertible symmetries}},  {\em JHEP} {\bf 08} (2025) 009,
  [\href{http://arxiv.org/abs/2503.20863}{{\tt arXiv:2503.20863}}].

\bibitem{Evans:fusion}
D.~E. Evans and C.~Jones, {\it {An operator algebraic approach to fusion
  category symmetry on the lattice}},
  \href{http://arxiv.org/abs/2507.05185}{{\tt arXiv:2507.05185}}.

\bibitem{Harlow:disjoint-additivity}
D.~Harlow, S.-H. Shao, J.~Sorce, and M.~Srivastava, {\it {Disjoint additivity
  and local quantum physics}},  \href{http://arxiv.org/abs/2509.03589}{{\tt
  arXiv:2509.03589}}.

\bibitem{vanLuijk:cones}
L.~van Luijk, A.~Stottmeister, and H.~Wilming, {\it {Uniqueness of
  Purifications Is Equivalent to Haag Duality}},  {\em Phys. Rev. Lett.} {\bf
  136} (2026), no.~4 040203, [\href{http://arxiv.org/abs/2509.12911}{{\tt
  arXiv:2509.12911}}].

\bibitem{Casini:DHR}
H.~Casini and J.~M. Magan, {\it {A generalization of the DHR theorem for higher
  form symmetries}},  \href{http://arxiv.org/abs/2511.21810}{{\tt
  arXiv:2511.21810}}.

\bibitem{Ogata:HD}
Y.~Ogata, D.~P{\'e}rez-Garc{\'\i}a, and A.~Ruiz-de Alarc{\'o}n, {\it {Haag
  Duality for 2D Quantum Spin Systems}},
  \href{http://arxiv.org/abs/2509.23734}{{\tt arXiv:2509.23734}}.

\bibitem{Kitaev:double}
A.~Y. Kitaev, {\it {Fault tolerant quantum computation by anyons}},  {\em
  Annals Phys.} {\bf 303} (2003) 2--30,
  [\href{http://arxiv.org/abs/quant-ph/9707021}{{\tt quant-ph/9707021}}].

\bibitem{Buerschaper_2013}
O.~Buerschaper, J.~M. Mombelli, M.~Christandl, and M.~Aguado, {\it A hierarchy
  of topological tensor network states},  {\em Journal of Mathematical Physics}
  {\bf 54} (Jan., 2013) [\href{http://arxiv.org/abs/1007.5283}{{\tt
  arXiv:1007.5283}}].

\bibitem{Buerschaper:2010yf}
O.~Buerschaper, M.~Christandl, L.~Kong, and M.~Aguado, {\it {Electric-magnetic
  duality of lattice systems with topological order}},  {\em Nucl. Phys. B}
  {\bf 876} (2013) 619--636, [\href{http://arxiv.org/abs/1006.5823}{{\tt
  arXiv:1006.5823}}].

\bibitem{Yan_2022}
B.~Yan, P.~Chen, and S.~X. Cui, {\it Ribbon operators in the generalized kitaev
  quantum double model based on hopf algebras},  {\em Journal of Physics A:
  Mathematical and Theoretical} {\bf 55} (Apr., 2022) 185201,
  [\href{http://arxiv.org/abs/2105.08202}{{\tt arXiv:2105.08202}}].

\bibitem{Choi:axions}
Y.~Choi, H.~T. Lam, and S.-H. Shao, {\it {Non-invertible Gauss law and
  axions}},  {\em JHEP} {\bf 09} (2023) 067,
  [\href{http://arxiv.org/abs/2212.04499}{{\tt arXiv:2212.04499}}].

\bibitem{Sorce:notes}
J.~Sorce, {\it {Notes on the type classification of von Neumann algebras}},
  {\em Rev. Math. Phys.} {\bf 36} (2024), no.~02 2430002,
  [\href{http://arxiv.org/abs/2302.01958}{{\tt arXiv:2302.01958}}].

\bibitem{PhysRevD.11.395}
J.~Kogut and L.~Susskind, {\it Hamiltonian formulation of wilson's lattice
  gauge theories},  {\em Phys. Rev. D} {\bf 11} (Jan, 1975) 395--408.

\bibitem{PhysRevD.19.3715}
D.~Horn, M.~Weinstein, and S.~Yankielowicz, {\it Hamiltonian approach to $z(n)$
  lattice gauge theories},  {\em Phys. Rev. D} {\bf 19} (Jun, 1979) 3715--3731.

\bibitem{Buerschaper_2009}
O.~Buerschaper and M.~Aguado, {\it Mapping kitaev’s quantum double lattice
  models to levin and wen’s string-net models},  {\em Physical Review B} {\bf
  80} (Oct., 2009) [\href{http://arxiv.org/abs/0907.2670}{{\tt
  arXiv:0907.2670}}].

\bibitem{chung2025spontaneouslybrokennoninvertiblesymmetries}
K.~T.~K. Chung, U.~Borla, A.~H. Nevidomskyy, and S.~Moroz, {\it Spontaneously
  broken non-invertible symmetries in transverse-field ising qudit chains},
  \href{http://arxiv.org/abs/2508.11003}{{\tt arXiv:2508.11003}}.

\bibitem{Fulton:book}
W.~Fulton and J.~Harris, {\em Representation theory: a first course}.
\newblock {Springer Science \& Business Media}, 2013.

\bibitem{AliAhmad:2025bnd}
S.~Ali~Ahmad, M.~S. Klinger, and Y.~Wang, {\it {The Many Faces of
  Non-invertible Symmetries}},  \href{http://arxiv.org/abs/2509.18072}{{\tt
  arXiv:2509.18072}}.

\bibitem{Benini:2025lav}
F.~Benini, P.~Calabrese, M.~Fossati, A.~H. Singh, and M.~Venuti, {\it
  {Entanglement asymmetry for higher and noninvertible symmetries}},
  \href{http://arxiv.org/abs/2509.16311}{{\tt arXiv:2509.16311}}.

\bibitem{TURAEV1992865}
V.~Turaev and O.~Viro, {\it State sum invariants of 3-manifolds and quantum
  6j-symbols},  {\em Topology} {\bf 31} (1992), no.~4 865--902.

\bibitem{kirillov2010turaevviroinvariantsextendedtqft}
A.~Kirillov, Jr. and B.~Balsam, {\it {Turaev-Viro invariants as an extended
  TQFT}},  \href{http://arxiv.org/abs/1004.1533}{{\tt arXiv:1004.1533}}.

\bibitem{Kawagoe:2024tgv}
K.~Kawagoe, C.~Jones, S.~Sanford, D.~Green, and D.~Penneys, {\it {Levin-Wen is
  a Gauge Theory: Entanglement from Topology}},  {\em Commun. Math. Phys.} {\bf
  405} (2024), no.~11 266, [\href{http://arxiv.org/abs/2401.13838}{{\tt
  arXiv:2401.13838}}].

\bibitem{thorngren2019fusioncategorysymmetryi}
R.~Thorngren and Y.~Wang, {\it {Fusion category symmetry. Part I. Anomaly
  in-flow and gapped phases}},  {\em JHEP} {\bf 04} (2024) 132,
  [\href{http://arxiv.org/abs/1912.02817}{{\tt arXiv:1912.02817}}].

\bibitem{Drinfeld:1986in}
V.~G. Drinfeld, {\it {Quantum groups}},  {\em Zap. Nauchn. Semin.} {\bf 155}
  (1986) 18--49.

\bibitem{balsam2012kitaevslatticemodelturaevviro}
B.~Balsam and A.~Kirillov, Jr., {\it {Kitaev's Lattice Model and Turaev-Viro
  TQFTs}},  \href{http://arxiv.org/abs/1206.2308}{{\tt arXiv:1206.2308}}.

\bibitem{Lin_2021}
C.-H. Lin, M.~Levin, and F.~J. Burnell, {\it {Generalized string-net models: A
  thorough exposition}},  {\em Phys. Rev. B} {\bf 103} (2021), no.~19 195155,
  [\href{http://arxiv.org/abs/2012.14424}{{\tt arXiv:2012.14424}}].

\bibitem{milnor1965structure}
J.~W. Milnor and J.~C. Moore, {\it On the structure of hopf algebras},  {\em
  Annals of Mathematics} {\bf 81} (1965), no.~2 211--264.

\bibitem{Chang_2014}
L.~Chang, {\it Kitaev models based on unitary quantum groupoids},  {\em Journal
  of Mathematical Physics} {\bf 55} (Apr., 2014)
  [\href{http://arxiv.org/abs/1309.4181}{{\tt arXiv:1309.4181}}].

\bibitem{kassel1995quantum}
C.~Kassel, {\em Quantum Groups}, vol.~155 of {\em Graduate Texts in
  Mathematics}.
\newblock Springer-Verlag, New York, 1995.

\bibitem{montgomery1993hopf}
S.~Montgomery, {\em Hopf Algebras and Their Actions on Rings}, vol.~82 of {\em
  CBMS Regional Conference Series in Mathematics}.
\newblock American Mathematical Society, Providence, RI, 1993.

\bibitem{klimyk1997quantum}
A.~Klimyk and K.~Schm{\"u}dgen, {\em Quantum Groups and Their Representations}.
\newblock Springer, Berlin, Heidelberg, 1997.

\bibitem{Inamura:2021szw}
K.~Inamura, {\it {On lattice models of gapped phases with fusion category
  symmetries}},  {\em JHEP} {\bf 03} (2022) 036,
  [\href{http://arxiv.org/abs/2110.12882}{{\tt arXiv:2110.12882}}].

\bibitem{lu2026generalizedkramerswannierselfdualityhopfising}
D.-C. Lu, A.~Chatterjee, and N.~Tantivasadakarn, {\it Generalized
  kramers-wannier self-duality in hopf-ising models},
  \href{http://arxiv.org/abs/2602.10183}{{\tt arXiv:2602.10183}}.

\bibitem{cordova2024representationtheorysolitons}
C.~Cordova, N.~Holfester, and K.~Ohmori, {\it {Representation theory of
  solitons}},  {\em JHEP} {\bf 06} (2025) 001,
  [\href{http://arxiv.org/abs/2408.11045}{{\tt arXiv:2408.11045}}].

\bibitem{larson}
R.~G. Larson and D.~E. Radford, {\it Semisimple cosemisimple hopf algebras},
  {\em American Journal of Mathematics} {\bf 110} (1988), no.~1 187--195.

\end{thebibliography}\endgroup

\end{document}